\documentclass[a4paper,12pt]{article}
\usepackage{bm}
\usepackage{amsmath}
\usepackage{amssymb}
\usepackage[dvips]{graphicx}
\usepackage{pdflscape}
\usepackage{multirow}
\usepackage[sort]{natbib}
\usepackage{soul}
\usepackage{amsmath}
\usepackage{graphicx}
\usepackage{enumerate}
\usepackage{natbib}
\usepackage[hyphens]{url} 
\usepackage[utf8]{inputenc}
\usepackage{mathtools}
\usepackage{amsbsy}
\usepackage{amssymb}
\usepackage{color}
\usepackage{ulem}
\usepackage{placeins}
\usepackage{booktabs}
\usepackage{multirow}
\usepackage{hyperref}
 \usepackage{booktabs,subcaption,amsfonts,dcolumn}
\usepackage{xcolor,colortbl}
\usepackage{subcaption}
\usepackage[textfont=it]{caption}
\usepackage{verbatim}
\usepackage{chngcntr}

\usepackage{amsthm}
\newtheorem{thm}{Theorem}[section]

\newtheorem{lem}[thm]{Lemma}

\newtheorem{rem}[thm]{Remark}
\newtheorem{defi}[thm]{Definition}

\newtheorem{eg}[thm]{Example}

\definecolor{Gray}{gray}{0.85}
\definecolor{LightCyan}{rgb}{0.88,1,1}

\newcolumntype{a}{>{\columncolor{red}}c}
\newcolumntype{g}{>{\columncolor{green}}c}

\newcommand{\mb}[1]{\mathbf{#1}}

\newcommand{\inD}{\operatorname{\stackrel{\cD}{\to}}}

\newcommand{\rank}{\operatorname {rank}}
\newcommand{\NN}{{\mathbb N}}%
\newcommand{\RR}{{\mathbb R}}%
\newcommand{\SSS}{{\mathbb S}}%

\newcommand{\cD}{\mathcal{D}}

\newcommand{\fS}{\mathfrak{S}}

\makeatletter
\def\blfootnote{\xdef\@thefnmark{}\@footnotetext}
\makeatother

\title{Constrained Shape Analysis with Applications to RNA Structure
}
\author{Kanti V. Mardia$^{*,1}$, Benjamin Eltzner$^{\dagger}$ and Stephan F. Huckemann$^{\dagger,2}$}
\begin{document}
\maketitle

\blfootnote{$^{*}$ University of Leeds}
\blfootnote{$^{\dagger}$ Felix-Bernstein-Institute, University of G\"ottingen}
\blfootnote{$^{1}$ k.v.mardia@leeds.ac.uk}
\blfootnote{$^{2}$ huckeman@math.uni-goettingen.de}

\begin{abstract}
  In many applications of shape analysis, lengths between some landmarks are constrained. For instance, biomolecules often have some bond lengths and some bond angles constrained, and variation occurs only along unconstrained bonds and constrained bonds' torsions where the latter are conveniently modelled by dihedral angles. Our work has been motivated by low resolution biomolecular chain RNA where only some prominent atomic bonds can be well identified. Here, we propose a new modelling strategy for such constrained shape analysis starting with a product of polar coordinates (polypolars), where, due to constraints, for example, some radial coordinates should be omitted, leaving products of spheres (polyspheres). We give insight into these coordinates for particular cases such as five landmarks which are motivated by a practical RNA application. We also discuss distributions for polypolar coordinates and give a specific methodology with illustration when the constrained size-and-shape variables are concentrated. There are applications of this in clustering and we give some insight into a modified version of the MINT-AGE algorithm.
\end{abstract}
 
\section{Introduction}

Statistical Shape Analysis deals with analysing shapes of objects given landmarks (see for example, \cite{Dryden2016}) where some transformations are filtered out; in particular, for size-and-shape analysis, translation and rotation are filtered out. This paper is motivated by an application in predicting RNA structure (backbone conformations) where some distances (bond lengths) between landmarks are constant or constrained, see \cite{wiechers2025RNAPrecis}. We develop in this paper shape analysis under given constraints. One of the main approaches in shape analysis is to use Procrustes analysis for registration but here we find that the coordinates' approach of registration by frames is appropriate to allow for the known constraints. Here we deal with size-and-shape analysis but we indicate that the approach can be extended to other shapes such as similarity shapes and affine shapes.

In shape analysis of biomolecules, often atomic nuclei serve as landmarks. When some landmarks are constrained, due to fixed atomic bonds lengths and restricted bond angles, variation of shape thus occurs only along these bonds' torsions. Modelling these by dihedral angles results in a very convenient and concise shape description given by elements on a torus. As the torus is a product of one-dimensional spheres, in our approach, starting with \emph{polypolar coordinates} (products of polar coordinates of varying dimensions) motivated by \emph{Goodall-Mardia coordinates} \citep{goodallkvm1991,goodallkvm1992,goodall1993multivariate}, we aim at discarding radial parts due to constraints, leaving the spherical parts, leading to polyspheres: products of spheres. In fact, our polypolar coordinates allow to model various extents of constraints, from no constraints at all in classical shape analysis, to strong constraints in the above mentioned torus model. 

In this contribution we focus on modelling length constrained size-and-shape. There are also angular constrained size-and-shapes, or both, such as dihedral angles. In passing, in our examples, we also touch the latter. We note that a lot of work has already been done (e.g. \cite{murray2003rna,Altis2008,SargsyanWrightLim2012,EltznerHuckemannMardia2018,zoubouloglou2022scaled,wiechers2023learningJRSSC}) for length and angular constrained angular shape, but not in this unified way we propose here.

Also, our polypolar coordinates conveniently allow to model multicentring which is a generalization of translation.

The polysphere part can be subjected to the rich body of polysphere analysis (e.g. \cite{HE_LASR15}), an extension of torus PCA that was originally developed for torus data (see above).

Correctly modelling RNA backbone conformations in X-ray crystallography and cryoEM experimental 3D structures is a long standing challenge in structural biology (see \cite{liebschner2019macromolecular}). In addressing the problem of predicting high detail RNA structure geometry from the information available in low resolution experimental maps of electron density, it is found that at low resolution ($\gtrapprox 2.5$\AA{}) the five atoms (including 3 of the backbone) of the RNA structure can be determined from electron density but all other backbone atom positions cannot be determined. In contrast, high resolution can determine all backbone atomic positions ($\leq 2$\AA{}). In consequence, high detailed conformations allow for modelling all constraints (length and angular) yielding dihedral angles only.

\begin{figure}[b]
  \centering
  \includegraphics[width=0.9\textwidth]{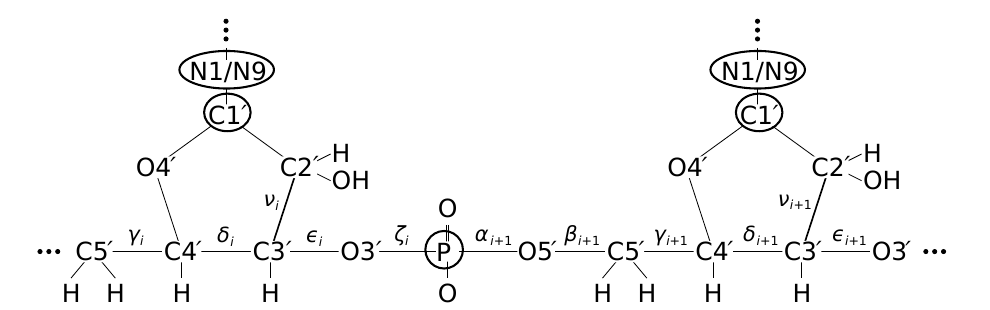}
  \caption{The five landmarks (circled) of the low detail RNA structure with the backbone and bases starting at N1/N9. \label{RNA}}
\end{figure}

As a running example we will specifically, consider the following \emph{five atoms} of an RNA suite, which ranges from one base to the next, 
\begin{center}
1: N1/N9, 2: C1', 3: P, 4: C1', 5: N1/N9\\ 
\end{center}
where P is Phosphorous, C1' is Carbon and N1/N9 is Nitrogen; carbon and nitrogen are on two sides of phosphorous. The bond lengths between 1 and 2, and between 4 and 5 are fixed, i.e. constrained. In contrast, the pseudo bond lengths between 2 and 3, as well as between 3 and 4 are variable. These are called \emph{pseudo bonds} because the participating atoms share chemical bonds only with atoms in between. P and the ones with a dash, namely C1', are on the backbone of the RNA whereas the atoms N1/N9 are part of the base of the RNA (see Figure \ref{RNA}). Note that the RNA backbone has many degrees of freedom: 7 dihedral angles, plus different pucker conformations of the ribose sugar, these are the pentagons to which bases bind, see Figure \ref{RNA}.

Achieving an accurate model of the RNA backbone is important for understanding RNA structure and function. We conclude with a brief illustration and propose a new method for mode hunting on a principal nested circle which is a core ingredient of the unsupervised RNA structure learning algorithm MINT-AGE, developed earlier \citep{wiechers2023learningJRSSC}.

\section{Size-and-Shape Coordinates}

\paragraph{Notation throughout this paper.} For $2 \leq k,m \in \NN$, the $k$-dimensional unit matrix is $I_k = \sum_{i=1}^k e_ie_i^T\in \RR^{k\times k}$ with the standard unit vectors $e_1,\ldots, e_k \in \RR^k$. Further, $1_k = e_1 + \ldots +e_k \in \RR^k$ has all entries equal to $1$, $0_k \in \RR^m$ is the vector with all entries equal to zero, and $SO(m)$ denotes the rotational group in $\RR^m$.

Given \emph{landmark} column vectors $x_j \in \RR^m$, $j=1,\ldots,k$, their \emph{configuration matrix} is $X = (x_1,\ldots,x_k) \in \RR^{m\times k}$.

Note that often in literature (see, for example, \cite{Dryden2016}) configuration matrices have landmarks as rows, then $X$ is a $k\times m$ matrix. Here, however, we find its transpose more convenient.

If $X$ is a configuration matrix, then a \emph{rigid body motion} (orientation preserving Euclidean transformation) takes $X$ to $ RX + a 1_k^T = X^*$, say, where $a\in \RR^m$ is an $m$-dimensional translation vector and $R \in SO(m)$ is an $m \times m$ rotation matrix. For the purposes of this paper, we then say that $X$ and $X^*$ have the same \emph{size-and-shape} following the standard terminology in shape analysis (see, for example, \citet{Dryden2016}). That is, the size-and-shape of $X$ is the equivalence class of configurations under the group of rigid motions. Recall that the term ``shape'' is a standard terminology in statistical shape analysis which means the equivalence class under the larger group of similarity transformations so the scale is also filtered out.

The translation and rotation effects can be filtered out by restricting an appropriate number of landmarks' coordinates to an appropriate subspace of the configuration space leading to \emph{size-and-shape coordinates}. Inspired by the literature and recent new applications we now introduce a broad scheme to obtain such coordinates.

We proceed in three steps: multicentring, polypolar coordinates and including constraints.

\subsection{Multicentring}
 
We assume that all configuration matrices $X = (x_1,\ldots,x_k)\in \RR^{m\times k}$ of concern feature a common \emph{frame}, i.e. $m+1$ landmarks in \emph{general position}. Thus, we assume that there is a fixed permutation $\pi\in \fS_k$ with fixed permutation matrix
$$P = \sum_{j=1}^k e_{\pi(j)} e_j^T$$
such that the first $m+1$ landmarks of
$$ (y_1,\ldots,y_k)=Y=XP_{\pi} = (x_{\pi(1)},\ldots,x_{\pi(k)})\,$$
are in general position, i.e. 
$y_2-y_1,\ldots,y_{m+1} -y_1$
are linearly independent. Then, \emph{general multicentring} is conveyed by a suitable fixed matrix $A \in \RR^{k\times k}$ of rank $k-1$ leading to landmarks
$$ W=(z_0,z_1,\ldots,z_{k-1}) := YA \mbox{ with }z_0=0_m\mbox{ and }z_1,\ldots,z_m\mbox{ linearly independent}\,.$$

Note that with a a singular value decomposition $A=\sum_{j=1}^{k-1} \lambda_j u_j v_j^T$ is with $U =(u_1,\ldots,u_k),V=(v_1,\ldots,v_k)\in SO(k)$ and $\lambda_1\geq,\ldots,\geq \lambda_{k-1} > 0$ from the multicentred configuration matrix
$$ W =YA$$
and the \emph{location information}
$$ Y u_k\,$$
removed by multicentring, the original configuration matrix can be retrieved:
$$ Y =(W + Yu_kv_k^T) (A + u_kv_k^T)^{-1}\,.$$

Of special interest and most convenient are multicentring matrices $A$ of form
\begin{eqnarray}
  \label{eq:general-mucen-A}
  A &=& I_k - \sum_{j=1}^k e_j \sum_{i=1}^k\varepsilon_{ij} e_i^T
  \quad \mbox{ with }\varepsilon_{ij} 
  \left\{\begin{array}{lcl} =1 &\mbox{for}& i,j = 1\,,\\
  =0&\mbox{for}& 2 \leq i =j \leq k\,,\\
  \in \{0,1\}&\mbox{ else.} 
  \end{array}\right.
\end{eqnarray}
Then, in $YA$ from all original landmarks $y_j$, the \emph{centre} $y_i$ is subtracted if $\varepsilon_{ij}\neq 0$, $1\leq i,j \leq k$. By design $y_1$ centred by itself, i.e. mapped to the origin, but no other landmark is centred by itself (otherwise too much location information would be lost). 

In fact, the number of $\{j \in \{1,\ldots,k\}: \exists 1\leq i \leq k \mbox{ with }\epsilon_{ij} \neq 0\}$ gives the number of centres. 

\begin{figure}
  \centering
  \includegraphics[width=1\linewidth]{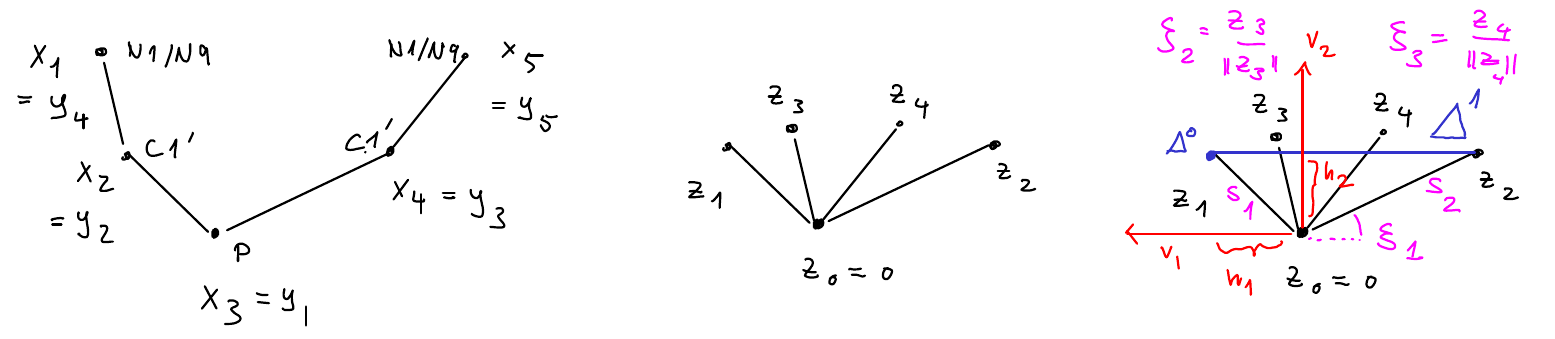}
  \caption{Coordinates for the five landmarks from low detail RNA backbone, see Figure \ref{RNA}. Left: Permuting landmarks. Middle: multicentring. Right: constrained MPP coordinates. See Examples \ref{eg:multicentring} and \ref{eg:MUCEPOPES}. \label{fig:5-MUCEPOPOS}}
\end{figure}

\begin{defi}\label{def:multicentring}
  Multicentring of a configuration matrix $X\in \RR^{m\times k}$ is achieved by first permuting the landmarks (columns) by a $k$-dimensional permutation matrix $P$, such that the first $m+1$ permuted landmarks are in general position, followed by multicentring the resulting landmarks by a matrix $A\in \RR^{k \times k}$ of form (\ref{eq:general-mucen-A}) and of $\rank(A) = k-1$. Then the columns $z_1,\ldots,z_{k-1}$ of
  $$ Z = XPA\left(\begin{array}{cc}
     0_{k-1}^T \\ \hline I_{k-1}
  \end{array}\right)\,.$$
  are the \emph{multicentred coordinates} of $X$ with origin $z_0 =0_m$.
\end{defi}

\begin{eg}\label{eg:multicentring}
  The following $A$ are frequently used.
  \begin{enumerate}
    \item \emph{Goodall-Mardia size-and-shape coordinates, Type 1} or GM SaS Type 1 coordinates use $\pi=id$ and $A=I_k - e_1 1_k^T$ and thus shift the first landmark to the origin.
    \item The \emph{five-point-representation}, see Figures \ref{RNA} and \ref{fig:5-MUCEPOPOS} from \cite{wiechers2025RNAPrecis} uses
    $$\pi = \left(\begin{array}{ccccc} 1 & 2 & 3 & 4 & 5\\ 3 & 2 & 4 & 1 &5 \end{array}\right)\mbox{ and }A=I_5 - e_1(e_1+e_2+e_3)^T -e_2e_4^T - e_3e_5^T\,.$$
    Here, the position $y_1=x_3$ of the phosphorous P together with the two C1' atom positions $y_2=x_2$ and $y_3=x_4$ naturally constitute a frame centred at $x_3$ which is subtracted from the other two atom positions yielding the pseudo bond lengths $\|z_1\| = \|y_2-y_1\| = \|x_2-x_3\|$ and $\|z_2\| = \|y_3-y_1\| = \|x_4-x_3\|$. Moreover from the two N1/N9 atom positions $y_4=x_1$ and $y_5=x_5$ their respective neighbouring C1' atom positions are subtracted, as the corresponding bond lengths $\|z_4\| = \|y_4-y_2\| = \|x_1-x_2\|$ and $\|z_5\| = \|y_5-y_4\| = \|x_5-x_4\|$ are fixed.
  
    Notably, here we essentially use three centres: $y_1=x_3, y_2=x_2$ and $y_3=x_4$.
  
    \item Dihedral angle representations (for strong constraints modelling, e.g. \cite{EltznerHuckemannMardia2018}, frequently use $\pi =id$ and
    $$A = I_k - e_1e_1^T - \sum_{j=1}^{k-1} e_j e_{j+1}\,.$$
    Here, all bond lengths $\|z_j\| = \|x_{j+1} - x_j\|$ for $1\leq j \leq k-1$ are fixed and so are the bond angles $\theta_j = \arccos \frac{z_{j+1}^Tz_j}{\|z_{j+1}\|\,\|z_j\|}$, $1\leq j \leq k-2$, see Figure \ref{fig:multicentred-dihedrals}.
  \end{enumerate}
\end{eg}

\begin{figure}
  \centering
  \includegraphics[width=0.8\linewidth]{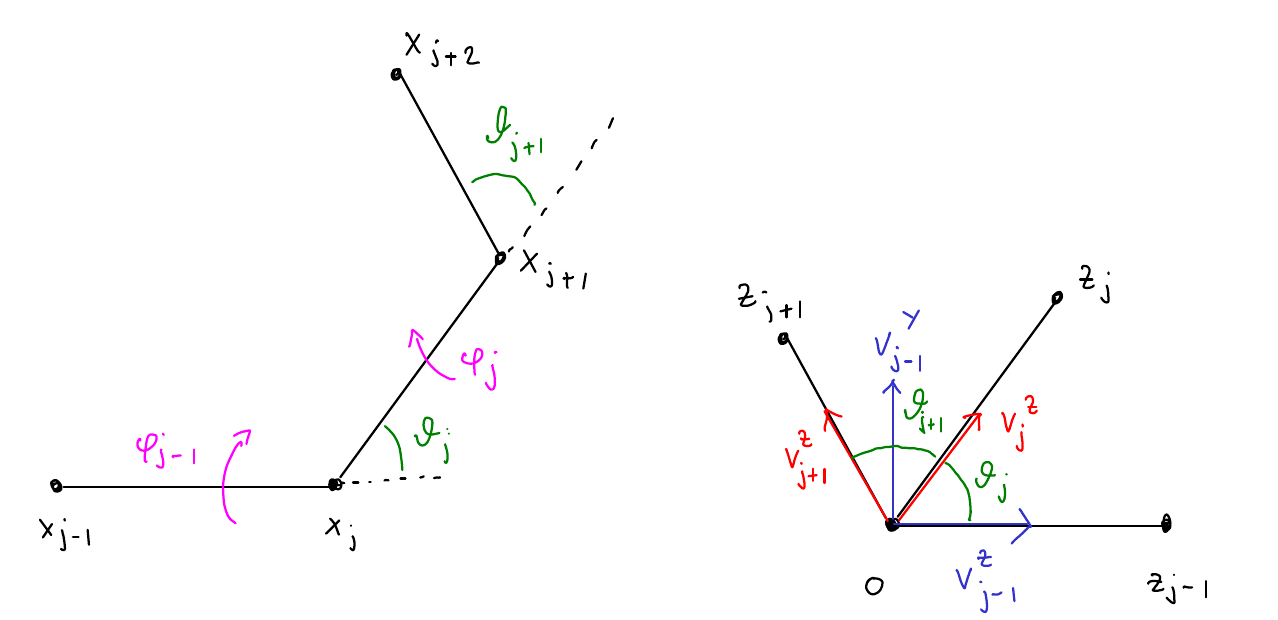}
  \caption{Left: part of a molecular backbone chain with bond angles $\theta_j,\theta_{j+1}$ and dihedral angles $\phi_{j-1},\phi_j$. Right: the same part multicentred with $z_{j-1} = x_j-x_{j-1},z_{j} = x_{j+1}-x_{j},z_{j+1} = x_{j+2}-x_{j+1}$. Their directions give the $z$-axis of the local coordinate system, where the $y$-axis is determined by the preceding $z$-coordinate. \label{fig:multicentred-dihedrals}}
\end{figure}

Notably, \emph{Goodall-Mardia size-and-shape coordinates, Type 2} or GM SaS Type 2 coordinates use $A=I_k - \frac{1}{k} 1_k 1_k^T$ which is the \emph{centring matrix} denoted by $C$ in \citet{Dryden2016}, subtracting the \emph{mean landmark}. This, however, does not fall into the paradigm of Definition \ref{def:multicentring}. Rather, removing translation was then achieved by isometrically and isomorphically mapping to the space of configurations of $k-1$ landmarks in $\RR^m$ via
$$ XB \in \RR^{m\times (k-1)}$$
with a matrix $B\in \RR^{k\times (k-1)}$ such that 
\begin{eqnarray}\label{eq:Helmertizing} \left(B\Big| \frac{1}{\sqrt{k} }\,1_k\right) &\in& SO(k)\,.\end{eqnarray}
Typically, a \emph{Helmert submatrix} is used for $B$.

We do not follow this avenue here because the new landmarks thus obtained loose their straightforward interpretation, as they are linear combinations of the original ones.  Further discussion can be found in \citet{Dryden2016}.

\subsection{Polypolar Coordinates}

In case of multicentring as in GM SaS Type 1 (see Example \ref{eg:multicentring}) or generalized multicentring as in GM SaS Type 2, \emph{Goodall-Mardia coordinates} of a configuration matrix $X$ are given by the columns of $T$ where 
$$ XA = VT$$ 
is a QR decomposition with $V\in SO(m)$ and upper triangular matrix $T$ \citep{goodallkvm1991,goodallkvm1992,goodall1993multivariate}. Inspired by this, we define multicentred polypolar coordinates.

\begin{defi}\label{def:popypolars}
  Suppose that $Z=(z_1,\ldots,z_{k-1})$ with $z_j \neq 0_m$ for all $1\leq j \leq k-1$ are multicentred coordinates of a configuration matrix $X$ obtained through $A$ as in Definition \ref{def:multicentring}. Then, with the QR decomposition
  \begin{equation*}\label{GM}
    Z= V T, 
  \end{equation*}
  where $V\in SO(m)$ and $T = (t_1,\ldots,t_{k-1})\in \RR^{m\times (k-1)}$ is an upper triangular matrix, determined by $t_{ii} >0$, where $t_i = (t_{i1},\ldots,t_{im})^T$, \emph{multicentred polypolar coordinates} (MPP coordinates) are given by
  $$\big((r_1,\zeta_1),\ldots,(r_{m-1},\zeta_{m-1}),\ldots,(r_{k-1},\zeta_{k-1})\big) \in \left(\prod_{j=1}^{m-1} (\RR_+\times \SSS^{j-1}_+)\right) \times \left( \RR_+\times \SSS^{m-1}\right)^{k-m}\,. $$
  Here, 
  $$r_j := \|t_j\|\mbox{ for }1\leq j\leq k-1\mbox{ and } \zeta_j := \frac{t_j}{r_j} \in \left\{\begin{array}{lcl}\SSS_+^{j-1} &\mbox{for}& 1\leq j \leq m-1\,,\\ \SSS^{m-1} &\mbox{for}& m\leq j \leq k-1\,.\end{array}\right.$$
  with the sphere 
  \begin{eqnarray*}
  \SSS^{m-1} &:=& \{\xi \in \RR^{m}: \|\xi\| = 1\}
  \end{eqnarray*}
  and the nested halfspheres,
  \begin{eqnarray*}
  \SSS^{j-1}_+ &:=& \{\zeta \in \SSS^{m-1}: e_{j}^T\zeta > 0,  e_{j+1}^T\zeta = \ldots = e_{m}^T\zeta =0\}
  \end{eqnarray*}
  for $1\leq j \leq m-1$.
\end{defi}

Note that $\SSS^0_+$ comprises the single point $e_1$, so that there is no variability along zero dimensional half spheres, which will thus be omitted.

Under constraints, some of the $r_j$ will be fixed and thus not considered any more. Also some of the spherical coordinates may be fixed under constraints. Then multicentred polypolar coordinates essentially live in an \emph{orthant} part times a \emph{polysphere} part:
$$ (\RR_+)^{k_1} \times \left(\prod_{j=1}^{k_2} S^{q_j}\right)$$
with $\RR_+ = (0,\infty)$ and spheres $S^{q_j}$ of dimension $1\leq q_j \leq m-1$, some of which may only be half spheres, $0\leq k_1\leq k-1$ and $1\leq j \leq k_2 \leq k-1$. Often, due to constraints, $k_1\ll k-1$ leading to a dominating polysphere part.

\begin{eg}\label{eg:MUCEPOPES} 
  In the five-point representation from Example \ref{eg:multicentring}, where $m=3$, we have the variable pseudo bond lengths $r_1 = \|z_1\|$ and $r_2 = \|z_2\|$, with $\zeta_1 = e_1$ omitted, as noted above, and the fixed bond lengths $\|z_3\|$ and $\|z_4\|$, also omitted. We have, however, different MPP coordinates depending on the embedding dimension.
 
  \begin{itemize}
    \item[(i)] In case of a planar configuration ($m=2$), see Figure \ref{fig:2D5-pts-F}, we have thus $\zeta_1,\zeta_2,\zeta_3 \in \SSS^1$ yielding the MPP coordinates
    $$(r_1,r_2,\zeta_1,\zeta_2,\zeta_3) \in \RR_+^2 \times (\SSS^1)^3\,. $$
    \item[(ii)] In case of a $m=3$ dimensional configuration, see Figure \ref{fig:3D5-pts-F} we have in contrast $\zeta_1\in \SSS^1_+$ and $\zeta_2,\zeta_3 \in \SSS^2$ yielding the MPP coordinates
    $$(r_1,r_2,\zeta_1,\zeta_2,\zeta_3) \in \RR_+^2 \times \SSS^1_+\times (\SSS^2)^2\,. $$
  \end{itemize}
\end{eg}

With more effort, we obtain dihedral angle representation from strongly constrained MPP coordinates.

\begin{eg}\label{eg:MUCEPOPES4dihedrals}
  Recall from Example \ref{eg:multicentring}, where $m=3$, multicentring a backbone chain $x_1,\ldots,x_k$ by $A = I_k -e_1e_1^T-\sum_{j=1}^{k-1} e_j e_{j+1}^T$ yielding
  $$Z=(z_1,\ldots,z_{k-1})=(x_2-x_1,\ldots,x_k-x_{k-1})\,.$$
  In case of fixed bond lengths $\|z_j\| = r_j>0$ ($1\leq j \leq k-1$) and fixed bond angles $\theta_j \in [0,2\pi)$ from $z_{j}$ to $z_{j+1}$, $\cos \theta_j = \frac{z_{j}^Tz_{j+1}}{r_jr_{j+1}}$ ($1\leq j \leq k-2$) only the variable torsion angles $\phi_j \in [0,2\pi)$ along the bond $z_{j}$ ($2\leq j \leq k-2$) convey shape variability, see Figure \ref{fig:multicentred-dihedrals}.

  We now introduce local coordinate systems. For $j=1$ we determine only a vertical axis 
  $$v_1^z := \frac{z_1}{r_1}\,,$$
  while for $2\leq j \leq k-1$ we define all three axes by 
  $$v_j^z := \frac{z_j}{r_j},\quad v_j^y := \frac{v^z_{j-1}-v_j^z(v_j^z)^Tv^z_{j-1} }{\|v^z_{j-1}-v_j^z(v_j^z)^Tv^z_{j-1} \|} = \frac{v^z_{j-1}-\cos\theta_{j-1}v_j^z}{\|v^z_{j-1}-\cos\theta_{j-1}v_j^z\|},\quad v_j^x := v_j^y \times v_j^z\,.$$
  In the last term, we used the three-dimensional vector cross product. Then, for $2\leq j \leq k-2$,
  $$v^z_{j+1} = \cos \theta_{j} v^z_{j} + \sin \theta_{j}\left(\cos \phi_j v^y_{j}+\sin \phi_j v^x_{j}\right)\,.$$
  and, if we restrict $0<\theta_j<\pi$, then
  $$\zeta_{j+1} = \frac{1}{\sin \theta_{j}}\left( (v^y_{j})^T v^z_{j+1}, (v^x_{j})^T v^z_{j+1}\right) \in \SSS^1\subset \SSS^2$$
  yielding the constrained MPP coordinates comprising only dihedral angles
  $$(\phi_j)_{j=2}^{k-2} = \left(\arctan_2\left( (v_j^x)^Tv_{j+1}^z,\,(v_j^y)^Tv_{j+1}^z\right)\right)_{j=2}^{k-2} \in (\SSS^1)^{k-3}\,.$$
\end{eg}

\subsection{Simplex MPP coordinates}

Here the frame $y_1,\ldots,y_{m+1}$ is viewed as a \emph{tower} of simplices \citep{barany1993random} which determine in a canonical way a new coordinate system given by $U\in SO(m)$ for the landmark columns of $T$ from Definition \ref{def:popypolars}. After multicentring via Definition \ref{def:multicentring}, the $z_1,\ldots,z_j$ span a simplex $\Delta^{j-1}\subset\RR^m$ of dimension $(j-1)$, $1\leq j \leq m$. Notably, these simplices form a tower, i.e. they are nested 
$$\{z_1\} = \Delta^0\mbox{ is a vertex of }\Delta^1\mbox{, which is an edge of }\Delta^2 \ldots \Delta^{m-2}\mbox{, which is a face of }\Delta^{m-1}\,.$$
Thus, choose the rotation $U=(u_1,\ldots,u_m)\in SO(m)$ such that $\SSS^{m-1}\ni u_{m} \perp \Delta_{m-1}$ pointing towards $\Delta_{m-1}$ and iteratively $\SSS^{m-1}\ni u_{m-j+1} \perp \Delta_{m-j}$ pointing towards $\Delta_{m-j}$ and $u_{m-j+1} \perp u_m,\ldots,u_{m-j+2}$ for $j=2,\ldots,m-1$.

The following lemma shows that the above approach and a closely related approach are possible and its constructive proof yields at once two implementable algorithms.

\begin{lem}\label{lem:simplex}
  Let $z_1,\ldots,z_m \in \RR^m$ be linearly independent. Then there are unique $U=(u_1,\ldots,u_m)\in SO(m)$ and
  \begin{enumerate}
    \item[(i)] unique $h_1\in \RR$, $h_2, \ldots,h_m >0$ such that
    $u_j^Tz_i = h_j\mbox{ for all } 1\leq i \leq j \leq m\,,$
    \item[(ii)] unique $h_1, \ldots,h_{m-1} >0 = h_m$ such that
    $ u_j^Tz_i = h_j\mbox{ for all } 1\leq i \leq j \leq m-1$ and $u_m^Tz_i = h_m$ for all $1\leq i \leq m-1$.
  \end{enumerate}
\end{lem}

\begin{proof}
  Existence: By hypothesis of linear independence, in case of (i), the system of linear equations
  \begin{eqnarray}\label{eq:proof1-lem-simplex}
    w^Tz_i = 1,&& 1\leq i \leq m 
  \end{eqnarray}
  has a unique solution $0_m\neq w \in \RR^m$. Set $u_m:= w/\|w\|$ and $h_m := \|w\|$.
  
  Similarly, in case of (ii), the system of linear equations 
  \begin{eqnarray}\label{eq:proof2-lem-simplex}
    w^Tz_i = 0,&& 1\leq i \leq m-1
  \end{eqnarray}
  is solved by a one dimensional linear subspace spanned by a unique $w\in \SSS^{m-1}$ with $w^Tz_m > 0$. We set $u_m:= \epsilon w$ and $h_m := 0$ where $\epsilon \in \{-1,1\}$ will be determined in the last step.
  
  Then, for both cases (i) and (ii), iteratively for $j =m-1,\ldots,2$, observe that the system of equations
  \begin{eqnarray}\label{eq:proof3-lem-simplex}
    &w^Tz_i = 1,~~ 1\leq i \leq j,\quad w^Tu_i=0,~~ j+1\leq i \leq m
  \end{eqnarray}
  has a unique solution $0_m\neq w \in \RR^m$. Set $v_j:= w/\|w\|$ and $h_j=\|w\|$. In case of (ii) also solve the above system for $j=1$ to obtain unique $u_1$ and $h_1>0$.

  In case of (i), $u_1$ is uniquely determined by requiring that $(u_1,\ldots,u_m) \in SO(m)$. In case of (ii), $\epsilon \in \{0,1\}$ is uniquely determined such that $(u_1,\ldots,u_{m-1},\epsilon u_m) \in SO(m)$, and then, write $u_m$ for $\epsilon u_m$.

  Uniqueness: Suppose that there were other $U'=(u'_1,\ldots,u'_m) \in SO(m)$ and $h'_2,\ldots,h'_m >0$ (Case (i)) or $h'_1,\ldots,h'_{m-1} >0$ (Case (ii)) with the asserted properties. Then, $w = u'_m/h'_m$ uniquely solves (\ref{eq:proof1-lem-simplex}) in case of (i), and solves (\ref{eq:proof2-lem-simplex}) in case of (ii), thus yielding $u'_m = u_m$ and $h_m = h'_m$ (Case (i)), and $u'_m = \epsilon' u_m$ with $\epsilon \in \{-1,1\}$ (Case (ii)). 

  Further, iteratively for $j=m-1,\ldots,2$, we have that $u'_j/h'_j$ solves (\ref{eq:proof2-lem-simplex}), and by uniqueness conclude $u'_j =u_j$, $h'_j= h_j$. Uniqueness extends to $j=1$ in case of (ii) and thus $\epsilon=1$ if $\det(U')=1$. In case of (i), similarly $u'_1 =u_1$ follows from $\det(U')=1$, completing the proof.
\end{proof}

\begin{rem}\label{rem:simplex}
  With $U=(u_1,\ldots,u_m) \in SO(m)$, $z_1,\ldots,z_m \in \RR ^m$ linearly independent and $h_1,h_2,\ldots,h_m$ from Lemma \ref{lem:simplex}, for $j=2,\ldots,m-2$,
  $$ h_j = u_j^Tz_j > 0$$
  is the \emph{height} (the distance from the origin) for the simplex $\Delta^{j-1}$ spanned by $z_1,\ldots,z_j$. The height of $\Delta^{0}$ spanned by $z_1$ is $|h_1|$ (Case (i)) and $h_1=u_1^Tz_1$ (Case (ii)). Only in Case (i), $h_m = u_m^Tz_m$ is the height of $\Delta^{m-1}$.

  For $j=1,\ldots, m-1$, for the following, we set 
  $$U_j :=(u_1,\ldots,u_j,0_m,\ldots,0_m)\in \RR^{m\times m}$$ 
  and note that 
  $\|U_j^Tz_j\| = \sqrt{\|z_j\|^2-h_{j+1}^2 - \ldots - h_m^2}$ for $j=1,\ldots, m-2$.
\end{rem} 
  
For multicentred coordinates $Z=(z_1,\ldots,z_{k-1})$, with the above $U$, in principle, simplex MPP coordinates are given by the columns of $U^TZ$. We distinguish two types.

\begin{defi}
  Suppose that $Z=(z_1,\ldots,z_{k-1})$ with $z_j \neq 0_m$ for all $1\leq j \leq k-1$ are multicentred coordinates of a configuration matrix $X$ obtained through $A$ as in Definition \ref{def:multicentring}. Moreover, assume that $U=(u_1,\ldots,u_m)\in SO(m)$ and $h_1,\ldots, h_m$ where $h_2,\ldots,h_{m-1} > 0$ and $h_1\in \RR$, $h_m>0$ (Case (i)) or $h_1>0$, $h_m=0$ (Case (ii)) have been obtained from $(z_1,\ldots,z_m)$ by Lemma \ref{lem:simplex}. With the notation of $U_j$ from Remark \ref{rem:simplex}, setting in Case (i),   
  $$\zeta_{j-1} := \left\{\begin{array}{rcl}
     U_j^Tz_j& \mbox{for}&2\leq j\leq m \\
     U^Tz_j& \mbox{for}& m+1\leq j\leq k-1
  \end{array}\right.\mbox{ and } 
  s_{j-1} := \left\{\begin{array}{rcl}
     \|U_j^Tz_j\|& \mbox{for}&2\leq j\leq m \\
     \|z_j\|& \mbox{for}& m+1\leq j\leq k-1
  \end{array}\right.
  \,,$$
  then
  \begin{align*}
    \big(h_1,(s_1,\zeta_1),&\ldots, (s_{m-1},\zeta_{m-1}),(s_{m},\zeta_{m}),\ldots,(s_{k-2},\zeta_{k-2}) \big)\\
    &\in \RR \times \prod_{j=1}^{m-1} (\RR_+ \times \SSS^j_+) \times (\RR_+ \times \SSS^{m-1})^{k-m-1}
  \end{align*} 
  are \emph{simplex MPP coordinates of Type 1}, 
  
  setting in Case (ii),    
  $$\zeta_{j-1} := \left\{\begin{array}{rcl}
      U_j^Tz_j& \mbox{for}&2\leq j\leq m-1 \\
      U^Tz_j& \mbox{for}& m\leq j\leq k-1
    \end{array}\right.\mbox{ and } 
    s_{j-1} := \left\{\begin{array}{rcl}
      \|U_j^Tz_j\|& \mbox{for}&2\leq j\leq m-1 \\
      \|z_j\|& \mbox{for}& m\leq j\leq k-1
    \end{array}\right.
  \,,$$
  then
  \begin{align*}
    \big(h_1,(s_1,\zeta_1),&\ldots, (s_{m-2},\zeta_{m-2}),(s_{m-1},\zeta_{m-1}),\ldots,(s_{k-2},\zeta_{k-2}) \big)\\
    &\in \RR_+ \times \prod_{j=1}^{m-2} (\RR_+ \times \SSS^j_+) \times (\RR_+ \times \SSS^{m-1})^{k-m-1}
  \end{align*} 
  are \emph{simplex MPP coordinates of Type 2}. Figure \ref{fig:polysphere-m-4} illustrates the the choice of $U$ for $m=4$.
\end{defi}

\begin{figure}
  \centering
  \includegraphics[width=0.45\linewidth]{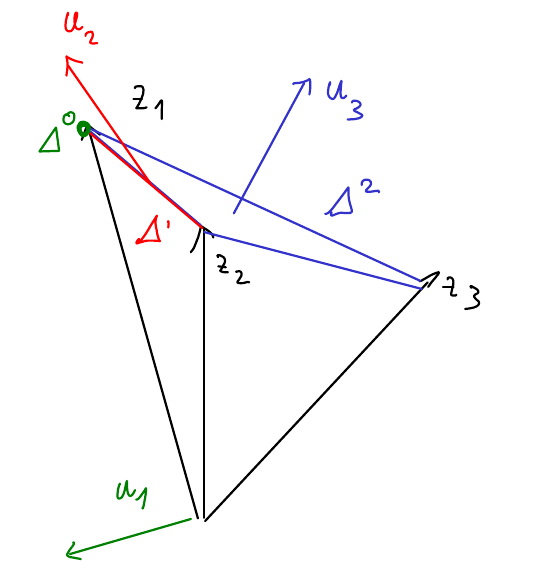}
  \caption{Simplex MPP coordinates in case of $m=4$. For Type 1 coordinates, $u_4$ is orthogonal to the simplex spanned by $z_1,\ldots,z_4$ (not depicted) and for Type 2 coordinates, $u_4$ is orthogonal to each of $z_1,z_2,z_3$. For both types $u_3$ is orthogonal to the simplex (triangle) $\Delta^2$ spanned by $z_1,z_2,z_3$, and $u_2$ is orthogonal to the simplex (linear segment) $\Delta^1$ spanned by $z_1,z_2$, determining $u_1$ such that $(u_1,\ldots,u_4) \in SO(4)$ (Type 1), and such that $u^T_1z_1>0$ (Type 2). \label{fig:polysphere-m-4}}
\end{figure}

\begin{rem}
  Due to Remark \ref{rem:simplex}, instead of $s_1,\ldots,s_{m-1}$ (Case (i)) equivalently the simplex heights $h_2,\ldots,h_m$ can be used and instead of $s_1,\ldots,s_{m-2}$ (Case (ii)) equivalently the simplex heights $h_2,\ldots,h_{m-1}$ can be used. 
\end{rem}

In conclusion, we give several examples of simplex MPP coordinates, first without constraints, then with.

\begin{eg}
  For simplicity in the first three examples below we chose the centring matrix $A=I_k-e_1e_1^T$.
  \begin{enumerate}
    \item Three points $x_1,x_2,x_3 \in \RR^2$ resulting in $z_0 = x_1-x_1, z_1 = x_2-x_1$ and $z_2=x_3-x_2$. For simplex MPP coordinates of Type 1 we have $ u_2 \perp z_2-z_1$ and consequently $u_1 := \epsilon \frac{z_2-z_1}{\| z_2-z_1\|}$ with suitable $\epsilon \in \{0,1\}$, yielding
    $$h_1= u_1^Tz_1,~ s_1=\|z_2\|,~ \zeta_1 = \left(\begin{array}{c} u_1^T\\ u_2^T \end{array}\right)\,\frac{z_2}{\|z_2\|}\mbox{ taking values in }\RR\times \RR_+ \times \SSS^1_+\,,$$
    while for Type 2 we have $u_2 \perp z_1$ and  consequently $u_1 := \frac{z_1}{\|z_1\|}$, yielding 
    $$h_1= u_1^Tz_1,~ s_1=\|z_2\|,~\zeta_1 = \left(\begin{array}{c} u_1^T\\ u_2^T \end{array}\right)\,\frac{z_2}{\|z_2\|}\mbox{ taking values in }\RR^2_+ \times \SSS^1\,.$$
    \item Four points $x_1,x_2,x_3,x_4 \in \RR^2$ resulting in $z_0 = x_1-x_1, z_j = x_{j+1}-x_1$ for $j=1,2,3$. For simplex MPP coordinates we have the  rotations $U=(u_1,u_2) \in SO(1)$ different (!) for Type 1 and Type 2, as above, now, respectively, adding the coordinates
    $$s_2 = \|z_3\|,~ \zeta_2 = \left(\begin{array}{c} u_1^T\\ u_2^T \end{array}\right)\,\frac{z_3}{\|z_3\|}\mbox{ taking values in } \RR_+ \times \SSS^1\,.$$
    \item Four points $x_1,x_2,x_3,x_4 \in \RR^3$ resulting in $z_0 = x_1-x_1, z_j = x_{j+1}-x_1$ for $j=1,2,3$. For simplex MPP coordinates of Type 1 we have $u_2, u_3 \perp z_2-z_1$ and $u_3 \perp z_3-z_2$. Hence, again, we have $u_1 := \epsilon \frac{z_2-z_1}{\|  z_2-z_1\|}$ with suitable $\epsilon \in \{0,1\}$, yielding
    \begin{align*}
      h_1&= u_1^Tz_1,~ &&s_1=\|U_2^Tz_2\| =\sqrt{\|z_2\|^2 - h_3^2}, &&\zeta_1 &= \left(\begin{array}{c}
      u_1^T\\ u_2^T \\0 \end{array}\right)\,\frac{z_2}{\|U_2^Tz_2\|}\,,\\ 
      s_2&=\|z_3\|\,,
      &&\zeta_2 = \left(\begin{array}{c}
      u_1^T\\ u_2^T \\u_3^T \end{array}\right)\,\frac{z_3}{\|z_3\|^2}\,,
    \end{align*}
    taking values in 
    $$\RR\times \RR_+ \times \SSS^1_+ \times \RR_+ \times \SSS^2_+ \,.$$
    For simplex MPP coordinates of Type 2 we have $u_3 \perp z_1,z_2$, $u_2 \perp z_2-z_1$. Again, consequently $u_1 := \epsilon \frac{z_2-z_1}{\|  z_2-z_1\|}$ with suitable $\epsilon \in \{0,1\}$, yielding
    $$h_1= u_1^Tz_1,~ s_1=\|z_2\|,~ \zeta_1 = \left(\begin{array}{c}
    u_1^T\\ u_2^T \\0 \end{array}\right)\,\frac{z_2}{\|z_2\|},~s_2=\|z_3\|,~ \zeta_2 = \left(\begin{array}{c}
    u_1^T\\ u_2^T \\u_3^T \end{array}\right)\,\frac{z_3}{\|z_3\|^2}\,,$$
    now taking values (note that $u_3^Tz_2 =0$ but $u_3^Tz_3\neq 0$) in 
    $$\RR_+\times \RR_+ \times \SSS^1_+ \times \RR_+ \times \SSS^2 \,.$$
    \item Recall the multicentred five-point-representation for low resolution RNA backbone from Example \ref{eg:multicentring}, see also Figures \ref{RNA} and \ref{fig:5-MUCEPOPOS}. In the following we consider $z_1,\ldots,z_4 \in \RR^m$ after multicentring from Example \ref{eg:multicentring}, for planar configurations ($m=2$) and proper three-dimensional configurations ($m=3$). Picking the different $U\in SO(m)$ and $h_1,\ldots,h_m$ as in the previous three examples, we arrive at the following.
    \begin{description}
      \item[$m=2$:] See Figures \ref{fig:4pts-in5-rep} and \ref{fig:2D5-pts-F}. From the unconstrained simplex MPP coordinates (having different values for Type 1 and Type 2) 
      $$h_1= u_1^Tz_1,~ s_j=\|z_{j+1}\|,~\zeta_j = \left(\begin{array}{c} u_1^T\\ u_2^T \end{array}\right)\,\frac{z_{j+1}}{\|z_{j+1}\|}\mbox{ for } j=1,2,3\,,$$ 
      taking into account that the bond lengths $s_2=\|z_3\|$ and $s_3=\|z_4\|$ are fixed and hence omitted, constrained simplex MPP coordinates of Type 1 take values in
      $$\RR \times \RR_+ \times \SSS_+^1 \times (\SSS^1)^2$$
      while the same of Type 2 take values in 
      $$\RR_+^2 \times (\SSS^1)^3\,.$$
      \item[$m=3$:] See Figure \ref{fig:3D5-pts-F}. From the unconstrained simplex MPP coordinates of Type 1,
      \begin{align*}
        h_1&= u_1^Tz_1,~ &&s_1=\|U_2^Tz_2\| =\sqrt{\|z_2\|^2 - h_3^2}, &&\zeta_1 = \left(\begin{array}{c}
        u_1^T\\ u_2^T \\0 \end{array}\right)\,\frac{z_2}{\|U_2^Tz_2\|}\,,\\ 
        s_j&=\|z_{j+1}\|\,,
        &&\zeta_j = \left(\begin{array}{c}
        u_1^T\\ u_2^T \\u_3^T \end{array}\right)\,\frac{z_{j+1}}{\|z_{j+1}\|^2}&&j=2,3\,,
      \end{align*}
      taking into account that the bond lengths $s_2=\|z_3\|$ and $s_3=\|z_4\|$ are fixed and hence omitted, constrained simplex MPP coordinates of Type 1 take values in
      $$(h_1,s_1,\zeta_1,\zeta_2,\zeta_3) \in \RR \times \RR_+ \times \SSS_+^1 \times \SSS^2_+ \times \SSS^2$$
      while the same of Type 2 take values in 
      \begin{eqnarray}\label{eq:sMMP-Type2}
        (h_1,s_1,\zeta_1,\zeta_2,\zeta_3) &\in& \RR_+^2 \times \SSS_+^1 \times (\SSS^2)^2\,.
      \end{eqnarray} 
    \end{description}
  \end{enumerate}    
  Note that in Type 1 coordinates $h_1$ takes the full range $(-\infty, \infty)$, while in Type 2 it is restricted to positive values. We will be concentrating on Type 2 only, which we will refer to as simplex MPP coordinates (unless there is any ambiguity) Further, \cite{wiechers2025RNAPrecis} utilize these latter constrained simplex MPP coordinates of Type 2.
\end{eg}

\section{Summary Statistics and  Data Analysis for RNA}\label{Summary}

\begin{figure}
  \centering
  \includegraphics[width=0.48\textwidth]{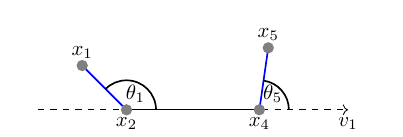}
  \includegraphics[width=0.48\textwidth]{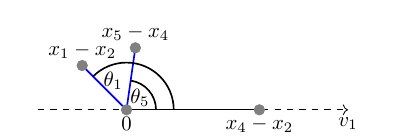}
  \caption{Left: four out of the five landmarks in 2D where $x_1$ and $x_5$ vary along circles of fixed radii about $x_2$ and $x_4$, respectively. Right: in multicentred polysphere coordinates. The two blue lines are of fixed constant lengths in both panels. \label{fig:4pts-in5-rep}}
\end{figure}

\begin{figure}
  \centering
  \includegraphics[width=1.0\textwidth]{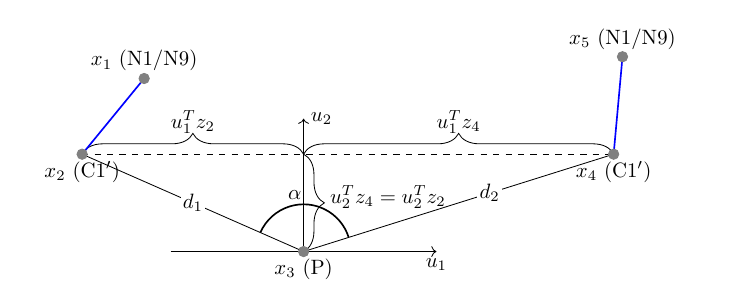}
  \caption{Turning 5 landmarks in a 2D plane into  multicentred (no common origin at $x_3$ and the base point of the blue lines yet depicted) generalized polycone coordinates. The two blue lines are of constant length. As this example is motivated by RNA suite analysis, we have included atom labels. \label{fig:2D5-pts-F}}
\end{figure}

\begin{figure}
  \centering
  \includegraphics{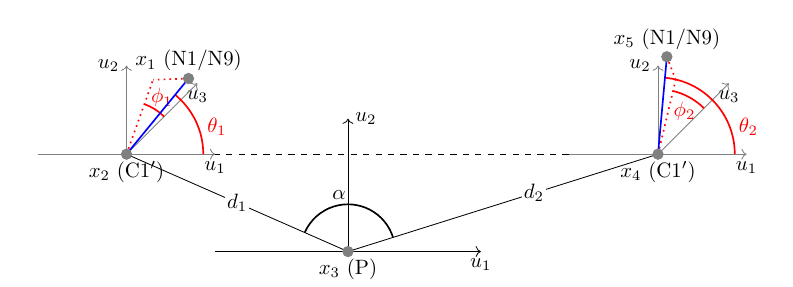}
  \caption{Similar to Figure \ref{fig:2D5-pts-F}, now turning 5 landmarks in a 3D space as occurring in low detail RNA representation into  multicentred (no common origin at $P$ and the base point of the blue lines yet depicted) polycone coordinates. Note the 2D simplex spanned by the 3 inner landmarks. The two blue lines are of constant length.  As this example is motivated by RNA suite analysis, we have included atom labels. \label{fig:3D5-pts-F}}
\end{figure}
We start with five landmarks so the Simplex MPP Type 2 coordinates are given by 
\begin{eqnarray}\label{eq:five-pt-coordinates}
  (h_1,s_1,\zeta_1,\zeta_2,\zeta_3) \in \RR_+\times (\RR_+ \times \SSS_+^1) \times (\SSS^2)^2.
\end{eqnarray}

We now write
$$h=d_1,s_1=d_2,\alpha,\zeta_2^T=\ell_1^T=(\theta_1,\phi_1),\;\zeta_3^T =\ell_2^T=(\theta_2, \phi_2)$$ 
Figure \ref{fig:3D5-pts-F} shows these. Note that the second and third variables $(d_2,\alpha)$ can be written in the Cartesian coordinates as 
 
\begin{equation}\label{eq:halfplane} 
  x=d_2 \cos \alpha, y=d_2\sin \alpha, 0<\alpha<\pi, 
\end{equation}
so it is a half-plane in 2D or the domain is $x>0, -\infty<y<\infty$. The rest are points on a two spheres in 3D. We thus have 7 variables
\begin{equation}
  W= (d_1,d_2, \alpha,\theta_1,\phi_1,\theta_2,\phi_2) \in (\RR_+^2, (0,\pi), (0,\pi), (0,2\pi),(0,\pi), (0,2\pi))^T .
\end{equation}

There is no known distribution of the random vector $W$ on such mixed variables. However, if we assume that data is concentrated (which we show we have in our example of the RNA data), we can make progress by using tangent approximations (see below, the Lambert azimuthal equal-area projection) for the spherical variables $\ell_1^T$ and $\ell_2^T$ where as the other three variables $d_1,d_2, \alpha$ are Euclidean as $\alpha$ is also so a linear variable not on the full circle.

We now introduce our variables for any given concentrated data as a multivariate normal vector $\mb{x}=(x_1, x_2, x_3, x_4, x_5, x_6, x_7)^T$ with 7 variables defined by 
\begin{equation}\label{7x}
  x_1=d_1, x_2=d_2, x_3=\alpha,x_4=\theta_1,x_5=\phi_1,x_6=\theta_2,x_7=\phi_2
\end{equation}
with now $x_4=\theta_1,x_5=\phi_1$ and $x_6=\theta_2,x_7=\phi_2$ are coordinates on the sphere after rotating the central direction, defined in equations \eqref{xbar} and \eqref{Rbar}, to the north pole, i.e. the $e_3$ direction.
 
Note that to assess whether a spherical data set is concentrated, we proceed as follows. Let ${\bf x}_1, \dots, {\bf x}_n, $ be $n$ points on $\SSS^2$. Then the location of these points can be summarised by their sample mean vector in $\RR^{3}$, which is
\begin{equation} \label{xbar}
  {\bf {\bar x}} = \frac{1}{n} \sum _{i=1}^n{\bf x} _i.
\end{equation}
Write the vector $\bar{\bf x}$ 

\begin{equation} \label{Rbar}
  {\bf {\bar x}} = {\bar R}{\bf {\bar x}}_0, \; 0\le {\bar R} \le 1,
\end{equation}
where ${\bf {\bar x}}_0$ is the sample {\bf mean direction} and ${\bar R}\; (= \lVert{\bf {\bar x}}\rVert)$ is the \textbf{mean resultant length}. Note that ${\bf {\bar x}}$ is the centre of gravity with direction ${\bf {\bar x}}_0$, and ${\bar R}$ is its distance from the origin. If ${\bar R}$ is nearly 1 then the data is concentrated and a formal test is normally carried out under the Fisher distribution which is a test of uniformity. The probability density function of the Fisher distribution in 3D as 
\begin{equation}\label{Fisher3D}
  f({\bf x}) = B(\kappa) \exp(\kappa \boldsymbol{\mu}^T {\bf x}), \quad {\bf x} \in \mathbb{R}^{3},\kappa \geq 0, \boldsymbol{\mu}^T\boldsymbol{\mu}=1, {\bf x}^T{\bf x}=1,
\end{equation}
where, 
$$B(\kappa) = \kappa/(2 \sinh \kappa),$$ 
$\kappa$ is the concentration parameter and $\mu$ is the population mean direction. The pdf is with respect to the Lebesgue measure. The uniformity test is based on $H_0: \kappa=0$ vs $H_0: \kappa \ge 0$. We can also calculate the $p-$values. For further details on the Fisher distribution and data analysis, see, for example, \citet{mardiajupp2000}.

A different avenue has been followed in \cite{wiechers2025RNAPrecis} where the Fr\'echet mean (a combination of Euclidean and intrinsic spherical minimization) has been computed for the 7-dimensional data. For concentrated data the following projection approach would yield very similar outcome.
 
{\bf Projection} Let $(\theta,\phi)$ be the spherical polar coordinates, then the Lambert azimuthal equal area projection (see, for example, \citet{mardiajupp2000}, p. 160) is given by 
\begin{equation}\label{Proj}
 (x,y)= (2\sin(\theta/2) \cos\phi, 2\sin(\theta/2) \sin\phi) 
\end{equation}
 
The data is first rotated so its mean direction is say north pole and carry out the analysis assuming these are drawn from a bivariate normal distribution. Also when plotted the data with the coordinates $(x,y)$ the sample mean direction is at the origin. A convenient choice of $(\theta '_i, \phi '_i)$ is as the spherical polar coordinates of is the rotation which takes the sample mean direction ${\bar x}_0$ to the north pole ${\bf n} = (0,0,1)^T$, (see,\citet{mardiajupp2000}). Let $(\theta ''_i, \phi ''_i)$ denote the spherical polar coordinates of $x_i$ in 
the coordinate system in which the sample mean direction ${\bar x}_0$ has spherical polar coordinates $(\theta '', \phi '') = (\pi /2, 0)$. Define the rotation matrix $A$ by 
$$
A = \left(\begin{array}{ccc}
\sin {\hat \alpha} \cos {\hat \beta} & 
\sin {\hat \alpha} \sin {\hat \beta} & \cos {\hat \alpha} \\
\sin {\hat \beta} & - \cos {\hat \beta} & 0 \\
\cos {\hat \alpha} \cos {\hat \beta} & 
\cos {\hat \alpha} \sin {\hat \beta} & - \sin {\hat \alpha} 
\end{array}\right),
$$
where
$$
{\bar x}_0 = \left(\begin{array}{c}
\sin {\hat \alpha} \cos {\hat \beta} \\
\sin {\hat \alpha} \sin {\hat \beta} \\
\cos {\hat \alpha} 
\end{array}\right).
$$
The spherical polar coordinates $(\theta ''_i, \phi ''_i)$ are defined by
$$
A x _i = \left(\begin{array}{c}
\sin \theta ''_i \cos \phi ''_i \\
\sin \theta ''_i \sin \phi ''_i \\
\cos \theta ''_i
\end{array}\right),
\label{LFthetphi''}
$$
with $\phi ''_i$ in the range $(- \pi, \pi ]$.

Let ${\bf x}_1, \dots, {\bf x}_n, $ be $n$ points on $S_2$ with the mean direction ${\bf {\bar x}}_0$ then we can rotate the data so that the data is centred about the mean direction. Let $A={\bf {\bar x}}_0 {\bf {\bar x}}_0^T$; let $B$ be the matrix of the eigenvectors of $A$ with the first eigenvector ${\bf b}_1$ proportional to ${\bf {\bar x}}_0$ and if $\quad {\bf {\bar x}}_0{\bf b}_1^T<0 \quad \text{then}\quad {\bf b}_1=-{\bf b}_1\quad \text {and if} \quad (\det(B)<0) \quad \text{then}\quad {\bf b}_2=-{\bf b}_2$ then the data is centred using $Y$= $XB$ where $X$ is the data matrix.

To get vector of longitudes and colatitudes after rotations of $X$, that is, from $Y$ $\phi=\arctan_2({\bf y}_3,{\bf y}_2), \theta  =\arccos({\bf y}_1)$ where $ Y=({\bf y}_1,{\bf y}_2,{\bf y}_3)$.

\subsection{RNA Example} 

Recall from the introduction that we have five landmarks (atoms) of the RNA Structure \\
1: N1/N9, 2: C1', 3: P, 4: C1', 5: N1/N9\\
where P is Phosphorous, C1' is Carbon and N1/N9 is Nitrogen; carbon and nitrogen are on two sides of phosphorous. The length between 1 and 2, and between 4 and 5 are fixed. The atoms P and C1' are on the backbone of the RNA whereas the atoms N1/N9 are part of the base of the RNA (see, Figure \ref{RNA} and Figure \ref{fig:3D5-pts-F}.)

We will analyse two illustrative RNA (Cluster 2 and Cluster 3 from P33) five-landmark data sets from low resolution ($\gtrapprox 2.5$\AA{}) experimental maps of electron density of the RNA structure (part of a cluster analysis, see our RNA paper \citet{wiechers2025RNAPrecis})

\begin{figure}[ht!]
  \centering
  \includegraphics[width=0.95\linewidth]{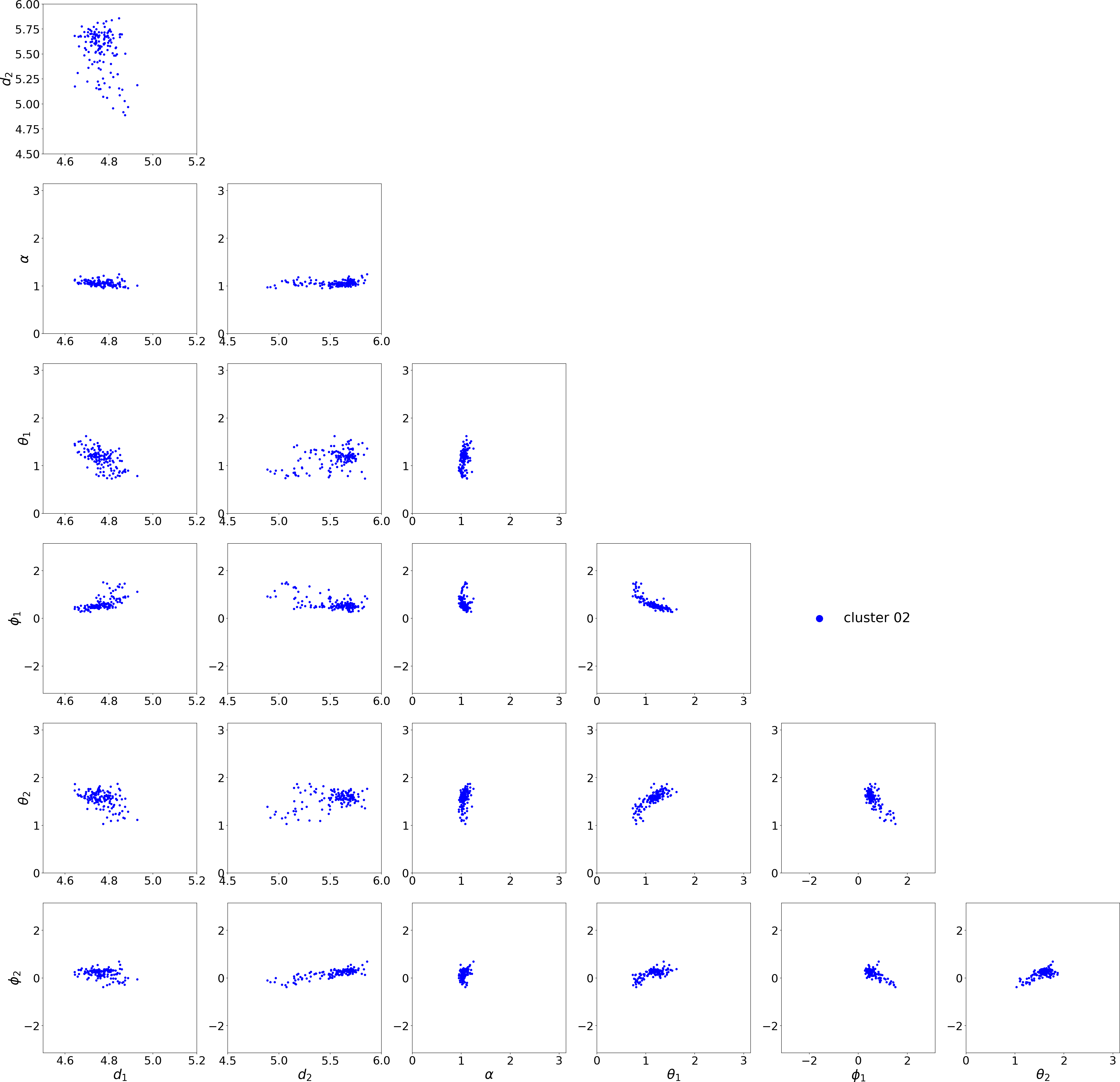}
  \caption{Scatter plots for Data 1 for seven Simplex MMP variables $W= (d_1,d_2, \alpha,\theta_1,\phi_1,\theta_2,\phi_2)$ with $n=146$. \label{fig:cluster2}}
\end{figure}

\textbf{Data 1 Analysis} We have the sample size $n$=146 for this data. Our $x_1, x_2$ are in \AA{}, $x_3$ is in radians (though in the original data, this was in degrees), $x_4,x_5,x_6,x_7$ original were in radians pre-projection. 

We give some relative indications of how this data is concentrated. We have 
$$\text{var}(x_1)=0.003, \text{var}(x_2)=0.045, \text{var}(x_3)=0.003. $$
Further, the mean resultant lengths of the spherical variables (pre-projections) are
$${\bar R}_1=0.960, {\bar R}_2=0.973 $$
so these are highly concentrated (we could apply the standard test of uniformity which would show that the hypothesis will be rejected with high confidence).

In fact,
$$\text{var}(x_4)=0.035, \text{var}(x_5)=0.064, \text{var}(x_6)=0.026, \text{var}(x_7)=0.030.$$
Further, each of the two original spherical data were rotated about their mean directions given respectively by 
$$ (0.73, 0.48, 0.40),\; ( 0.95, 0.20, 0.02).$$

We now give some more summary statistics with some general comments. The sample mean vector ${\bar x}$ of $x$ is given by 
$${\bar x}^T= (4.77, 5.54,1.06, 1.15,0.61,1.55,0.20).$$

Before investigating covariances, we represent the $\phi_j, \theta_j$ angles as points in $\mathbb{S}^2 \subset \mathbb{R}^3$, rotate such that the extrinsic mean is on the $x_1$ axis, and then calculate the Lambert azimuthal equal-area projection. The resulting sample covariance matrix $S$ of the thus transformed $x$ is given by 

\begin{align}
  \begin{split}
    S = \begin{pmatrix}
      \phantom{-}0.003 & -0.003 & -0.001 & \phantom{-}0.009 & \phantom{-}0.002 &-0.002 &\phantom{-}0.004 \\
      -0.003 & \phantom{-}0.045 & \phantom{-}0.002 & \phantom{-}0.032 &-0.005 & \phantom{-}0.028 & \phantom{-}0.017 \\
      -0.001 & \phantom{-}0.002 & \phantom{-}0.003 & -0.003 & -0.002 & \phantom{-}0.002 & -0.003 \\
      \phantom{-}0.009 & -0.032 & -0.003 & \phantom{-}0.071 & \phantom{-}0.017 & -0.028 & \phantom{-}0.037 \\
      \phantom{-}0.002 & -0.005 & -0.002 & \phantom{-}0.017 & \phantom{-}0.009 & -0.008 & \phantom{-}0.012 \\
      -0.002 & \phantom{-}0.028 & \phantom{-}0.002 & -0.028 & -0.008 & \phantom{-}0.025 & -0.018 \\
      \phantom{-}0.004 & -0.017 & -0.003 & \phantom{-}0.037 & \phantom{-}0.012 & -0.018 & \phantom{-}0.029
    \end{pmatrix}
  \end{split}
\end{align}

Note that the eigenvalues $0.134, 0.030, 0.010, 0.005, 0.003, 0.002, 0.001$ (which are scale dependent) lead to a determinant $|S|= 1.24 \times 10^{-15} $, which is very small but $S$ is not singular.

The sample correlation matrix $R$ of $x$ is given by 

\begin{align}
  \begin{split}
    COR = \begin{pmatrix}
      \phantom{-}1.00 & -0.30 & -0.24 & \phantom{-}0.66 & \phantom{-}0.37 & -0.25 & \phantom{-}0.44\\
      -0.30 & \phantom{-}1.00 & \phantom{-}0.19 & -0.56 & -0.25 & \phantom{-}0.82 & -0.46\\
      -0.24 & \phantom{-}0.19 & \phantom{-}1.00 & -0.22 & -0.35 & \phantom{-}0.20 & -0.31\\
      \phantom{-}0.66 & -0.56 & -0.22 & \phantom{-}1.00 & \phantom{-}0.66 & -0.67 & \phantom{-}0.81\\
      \phantom{-}0.37 & -0.25 & -0.35 & \phantom{-}0.66 & \phantom{-}1.00 & -0.55 & \phantom{-}0.72\\
      -0.25 & \phantom{-}0.82 & \phantom{-}0.20 & -0.67 & -0.55 & \phantom{-}1.00 & -0.67\\
      \phantom{-}0.44 & -0.46 & -0.31 & \phantom{-}0.81 & \phantom{-}0.72 & -0.67 & \phantom{-}1.00
    \end{pmatrix}
  \end{split}
\end{align}

Note that $|COR|= 0.0070 $ so that $COR$ is not singular but it is very small.
 
Since, the variables are in different units, the PCA on $R$ would be more information than of $S$. The eigenvalues of $R$ are 
$$ 3.92, 1.05, 0.84, 0.72, 0.26, 0.12, 0.09, $$
and the percentage contributions are
$$ 56.1, 15.0, 12.0, 10.3, 3.7, 1.7, 1.3, $$
so there is no real low dimensional reduction. The principal components are given by the following but there is no clear pattern.
\begin{align}
    \begin{split}
        \Gamma = \begin{pmatrix}
            \phantom{-}0.31 & \phantom{-}0.35 & -0.53 & \phantom{-}0.56 & -0.23 & -0.34 & \phantom{-}0.12\\
            -0.35 & \phantom{-}0.55 & -0.20 & -0.40 & \phantom{-}0.11 & \phantom{-}0.03 & \phantom{-}0.59\\
            -0.21 & -0.53 & -0.74 & -0.33 & -0.08 & -0.10 & -0.08\\
            \phantom{-}0.46 & \phantom{-}0.04 & -0.30 & \phantom{-}0.00 & \phantom{-}0.23 & \phantom{-}0.80 & \phantom{-}0.06\\
            \phantom{-}0.39 & \phantom{-}0.28 & \phantom{-}0.04 & -0.56 & -0.61 & -0.02 & -0.28\\
            -0.42 & \phantom{-}0.45 & -0.20 & \phantom{-}0.08 & \phantom{-}0.19 & \phantom{-}0.15 & -0.72\\
            \phantom{-}0.44 & \phantom{-}0.08 & -0.04 & -0.31 & \phantom{-}0.68 & -0.46 & -0.14
        \end{pmatrix}
    \end{split}
\end{align}

\begin{figure}[ht!]
  \centering
  \includegraphics[width=0.95\linewidth]{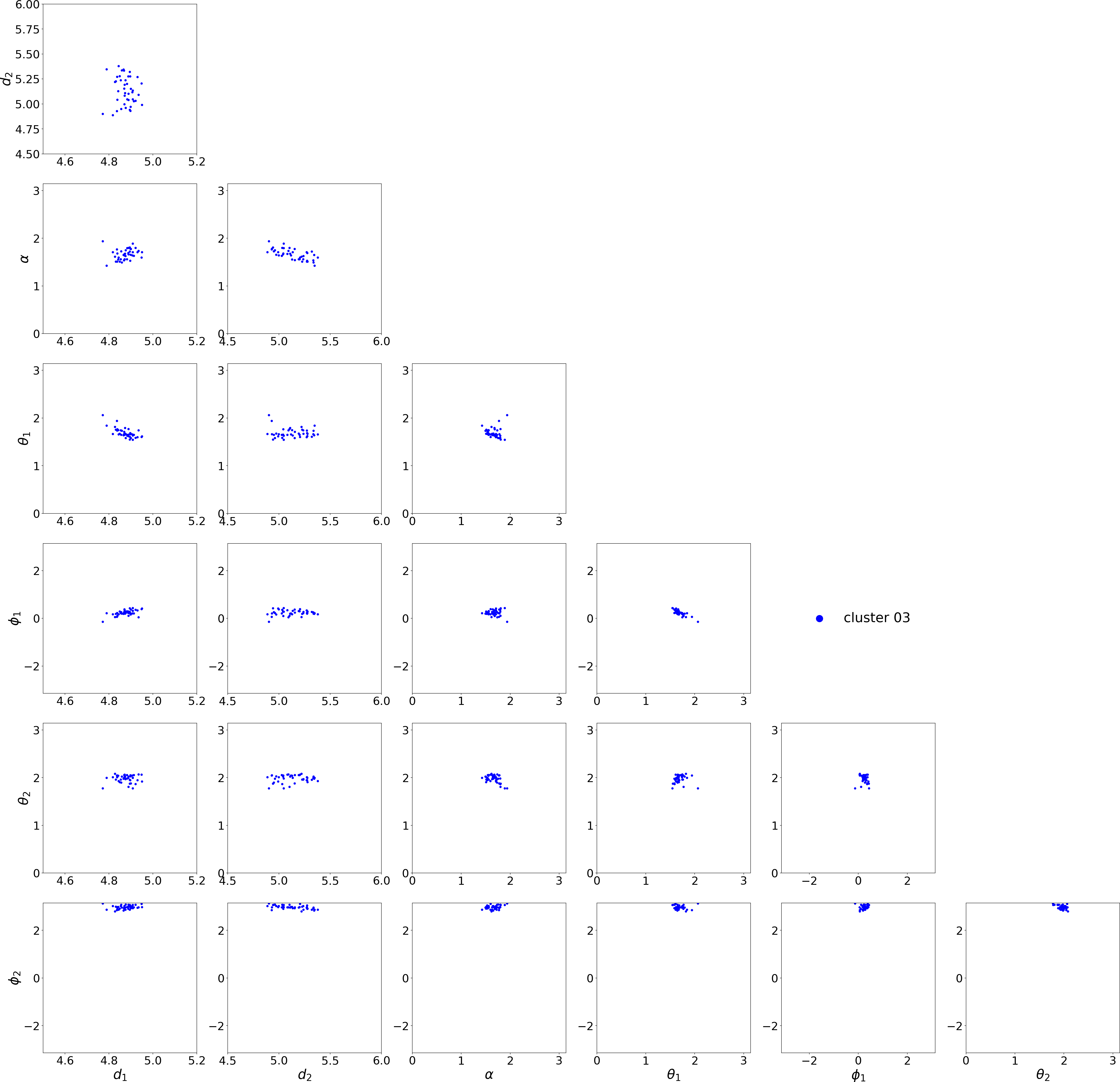}
  \caption{Scatter plots for Data 2 for seven Simplex MMP variables $W= (d_1,d_2, \alpha,\theta_1,\phi_1,\theta_2,\phi_2)$ with $n=44$.\label{fig:cluster3}}
\end{figure}

{\bf Data 2 Analysis} We have the sample size $n$=44 for this data. Our $x_1, x_2$ are in \AA{}, $x_3$ is in radians (though in the original data, this was in degrees), $x_4, x_5, x_6, x_7$ original were in radians pre-projection. 

We give some relative indications of how this data is concentrated. We have 
$$\text{var}(x_1)=0.001, \text{var}(x_2)=0.020, \text{var}(x_3)=0.012 $$
Further, the mean resultant lengths of the spherical variables (pre-projections) are
$${\bar R}_1=0.989, {\bar R}_2=0.995 $$
so these are highly concentrated. In fact,
$$\text{var}(x_4)=0.010, \text{var}(x_5)=0.012, \text{var}(x_6)=0.006, \text{var}(x_7)=0.006.$$
Further, each of the two original spherical data were rotated about their mean directions given respectively by 
$$ (0.95, 0.23, -0.12),\; ( -0.90, 0.16, -0.40).$$

We now give some more summary statistics with some general comments. The sample mean vector ${\bar x}$ of $x$ is given by 
$${\bar x}^T= (4.87, 5.13, 1.66, 1.69, 0.24, 1.98, 2.97).$$

Before investigating covariances, we represent the $\phi_j, \theta_j$ angles as points in $\mathbb{S}^2 \subset \mathbb{R}^3$, rotate such that the extrinsic mean is on the $x_1$ axis, and then calculate the Lambert azimuthal equal-area projection. The resulting sample covariance matrix $S$ of the thus transformed $x$ is given by

\begin{align}
  \begin{split}
    S = \begin{pmatrix}
      \phantom{-}0.001 & \phantom{-}0.000 & \phantom{-}0.001 & \phantom{-}0.001 & \phantom{-}0.003 & \phantom{-}0.000 & -0.001\\
      \phantom{-}0.000 & \phantom{-}0.020 & -0.011 & \phantom{-}0.000 & \phantom{-}0.001 & -0.003 & \phantom{-}0.004\\
      \phantom{-}0.001 & -0.011 & \phantom{-}0.013 & -0.001 & \phantom{-}0.000 & \phantom{-}0.005 & -0.001\\
      \phantom{-}0.001 & \phantom{-}0.000 & -0.001 & \phantom{-}0.005 & \phantom{-}0.006 & \phantom{-}0.001 & \phantom{-}0.000\\
      \phantom{-}0.003 & \phantom{-}0.001 & \phantom{-}0.000 & \phantom{-}0.006 & \phantom{-}0.017 & \phantom{-}0.001 & -0.001\\
      \phantom{-}0.000 & -0.003 & \phantom{-}0.005 & \phantom{-}0.001 & \phantom{-}0.001 & \phantom{-}0.008 & -0.001\\
      -0.001 & \phantom{-}0.004 & -0.001 & \phantom{-}0.000 & -0.001 & -0.001 & \phantom{-}0.004
    \end{pmatrix}
  \end{split}
\end{align}

Note that the eigenvalues $0.0299, 0.0207, 0.0079, 0.0047, 0.0023, 0.0011, 0.0004$ (which are scale dependent) lead to a determinant $|S|= 2.58 \times 10^{-17}$, which is very small but $S$ is not singular.

The sample correlation matrix $R$ of $x$ is given by 

\begin{align}
  \begin{split}
    COR = \begin{pmatrix}
      \phantom{-}1.00 & -0.07 & \phantom{-}0.23 & \phantom{-}0.37 & \phantom{-}0.68 & \phantom{-}0.04 & -0.23\\
      -0.07 & \phantom{-}1.00 & -0.67 & \phantom{-}0.04 & \phantom{-}0.07 & -0.27 & \phantom{-}0.53\\
      \phantom{-}0.23 & -0.67 & \phantom{-}1.00 & -0.07 & \phantom{-}0.03 & \phantom{-}0.54 & -0.14\\
      \phantom{-}0.37 & \phantom{-}0.04 & -0.07 & \phantom{-}1.00 & \phantom{-}0.69 & \phantom{-}0.18 & -0.08\\
      \phantom{-}0.68 & \phantom{-}0.07 & \phantom{-}0.03 & \phantom{-}0.69 & \phantom{-}1.00 & \phantom{-}0.11 & -0.09\\
      \phantom{-}0.04 & -0.27 & \phantom{-}0.54 & \phantom{-}0.18 & \phantom{-}0.11 & \phantom{-}1.00 & -0.25\\
      -0.23 & \phantom{-}0.53 & -0.14 & -0.08 & -0.09 & -0.25 & \phantom{-}1.00
    \end{pmatrix}
  \end{split}
\end{align}

Note that $|COR|= 0.0330 $ so that $R$ is not singular but it is very small.

Since, the variables are in different units, the PCA on $R$ would be more information than of $S$. The eigenvalues of $R$ are 
$$ 2.44, 2.01, 0.96, 0.79, 0.52, 0.19, 0.09, $$
and the percentage contributions are
$$ 34.9, 28.7, 13.7, 11.2, 7.4, 2.7, 1.3, $$
so there is no real low dimensional reduction. The principal components are given by the following but there is no clear pattern.
\begin{align}
    \begin{split}
        \Gamma = \begin{pmatrix}
            \phantom{-}0.42 & -0.32 & -0.17 & -0.52 & -0.39 & -0.43 & \phantom{-}0.28\\
            -0.38 & -0.45 & \phantom{-}0.25 & \phantom{-}0.10 & -0.53 & -0.12 & -0.53\\
            \phantom{-}0.41 & \phantom{-}0.39 & \phantom{-}0.36 & -0.40 & \phantom{-}0.10 & -0.05 & -0.61\\
            \phantom{-}0.33 & -0.45 & \phantom{-}0.13 & \phantom{-}0.41 & \phantom{-}0.53 & -0.46 & -0.14\\
            \phantom{-}0.39 & -0.50 & \phantom{-}0.04 & -0.07 & \phantom{-}0.03 & \phantom{-}0.76 & -0.05\\
            \phantom{-}0.37 & \phantom{-}0.23 & \phantom{-}0.54 & \phantom{-}0.47 & -0.44 & \phantom{-}0.01 & \phantom{-}0.32\\
            -0.35 & -0.19 & \phantom{-}0.68 & -0.40 & \phantom{-}0.29 & -0.01 & \phantom{-}0.36
        \end{pmatrix}
    \end{split}
\end{align}

{\bf Data 3 Analysis} As a third example, we consider the Mahalanobis distance and two-sample $T^2$ statistic (see, for example, \cite{MardiaKentTaylor2024}) for the two data sets Data 1 and Data 2 analysed above. Recall, for Data 1 for seven Simplex MMP variables $x^T= (d_1,d_2, \alpha,\theta_1,\phi_1,\theta_2,\phi_2)$ with $n=146$, the means are
$${\bar x}^T= (4.77, 5.54, 1.06, 1.15, 0.61, 1.55, 0.20)$$
and the corresponding variances are
$$ (0.003, 0.045, 0.003, 0.035, 0.064, 0.026, 0.030).$$

For Data 2, again for seven Simplex MMP variables $x^T= (d_1,d_2, \alpha,\theta_1,\phi_1,\theta_2,\phi_2)$ with $n=44$, the means are
$${\bar x}^T= (4.87, 5.13, 1.66, 1.69, 0.24, 1.98, 2.97)$$
and the corresponding variances are
$$ (0.001, 0.020, 0.012, 0.010, 0.012, 0.006, 0.006).$$

It can be seen that the means for Data 1 and Data 2 look similar except for the means of $\phi_2$ but also there is less variation in Data 2 though these all are small. In both cases, the variance for $d_1$ is smaller relative to the variance for $d_2$. Figures \ref{fig:cluster2} and \ref{fig:cluster3} indicate this behaviour though overall both data sets are concentrated and are different at least in the means.

One way to test the mean differences is to use the Hotelling $T^2$ test, which we now perform (assuming that the two population covariance matrices are equal). Since this requires a joint coordinate system, we perform the Lambert azimuthal equal-area projection on the pooled data set for the spherical coordinates as done above for the individual data sets.

This yields the mean vectors in centred coordinates
\begin{align*}
  \overline{y}_1 =&~ (-0.025, 0.093, -0.138, 0.130, 0.068, -0.182, 0.093)^T\\
  \overline{y}_2 =&~ (0.084, -0.315, 0.467, -0.435, -0.233, 1.651, -0.844)^T
\end{align*}

The pooled covariance matrix in the joint coordinates is
\begin{align}
  \begin{split}
    S_{\text{pooled}} = \begin{pmatrix}
      \phantom{-}0.003 & -0.003 & \phantom{-}0.000 & \phantom{-}0.008 & \phantom{-}0.002 & -0.003 & \phantom{-}0.002\\
      -0.003 & \phantom{-}0.040 & -0.001 & -0.024 & -0.006 & \phantom{-}0.025 & -0.001\\
      \phantom{-}0.000 & -0.001 & \phantom{-}0.005 & -0.002 & -0.001 & \phantom{-}0.003 & -0.001\\
      \phantom{-}0.008 & -0.024 & -0.002 & \phantom{-}0.056 & \phantom{-}0.017 & -0.031 & \phantom{-}0.015\\
      \phantom{-}0.002 & -0.006 & -0.001 & \phantom{-}0.017 & \phantom{-}0.010 & -0.012 & \phantom{-}0.006\\
      -0.003 & \phantom{-}0.025 & \phantom{-}0.003 & -0.031 & -0.012 & \phantom{-}0.032 & -0.007\\
      \phantom{-}0.002 & -0.001 & -0.001 & \phantom{-}0.015 & \phantom{-}0.006 & -0.007 & \phantom{-}0.015
    \end{pmatrix}
  \end{split}
\end{align}
where the overall low variances already indicate that the Mahalanobis distance will be large leading to a high value of the $T^2$ statistic and a miniscule p-value
\begin{align*}
  d_{\text{Mahalanobis}} := (\overline{y}_1 - \overline{y}_2)^T S_{\text{pooled}}^{-1} (\overline{y}_1 - \overline{y}_2) =&~ 443.4\\
  T^2 =&~ 14723.0\\
  p\text{-value} =&~ 1.8\cdot 10^{-168}\, .
\end{align*}
(Note that the value of the $F$ statistic from $T^2$ is found to be 2073. with degrees of freedom $f_1=7$, and $f_2= 182$, see, for example, \cite{MardiaKentTaylor2024}). This shows that the clusters can be clearly separated. These two examples give a hint that how multivariate analysis can be used for such a concentrated data,

\section{Distributions on Polypolar Coordinates}\label{Distr}
We have assumed in Section \ref{Summary} that if data are concentrated, we can get some insight using tangent projections. However, when the data are not concentrated, we need some plausible models. Note that the constrained size-and-shape coordinates are on a (half) polypolar space, some of which being spheres only, and we will give a method of construction of some plausible distributions on this space. There have been work in Directional Statistics for distributions on cylinders but not on polypolars. Our construction follows the Fisher approach where one take a multivariate normal distribution and conditioned it on the embedded space, in this case the polypolars.
 
We note the following two types of polar coordinates, the first for a full space, slicing with full spheres and another for a half space, slicing with half spheres. 

In the following we will be concerned with simplex MPP Type 2 coordinates. The corresponding results for Type 1 have similar form.
\begin{eqnarray}
  \lefteqn{\hspace*{-2cm} \nonumber
  \RR_+ \times \SSS^j = \{(r\cos \theta_1,r\sin\theta_1\cos \theta_2,\ldots, r\sin\theta_1\ldots \sin \theta_{j-1} \cos \theta_{j},r\sin\theta_1\ldots \sin\theta_j):}\\&&
  \label{eq:full-sphere-polar}
  \theta_l \in [0,\pi],1\leq l \leq j-1,\theta_j \in [0,2\pi), r>0\}\\ 
  \lefteqn{\hspace*{-2cm}
  \nonumber
  \RR_+ \times \SSS_+^j =\{(r\cos \theta_1,r\sin\theta_1\cos \theta_2,\ldots, r\sin\theta_1\ldots \sin \theta_{j-1} \cos \theta_{j},r\sin\theta_1\ldots \sin\theta_j):}\\
    \label{eq:half-sphere-polar}
  &&
  \theta_l \in [0,\pi/2],\theta_l \in [0,\pi],2\leq l \leq j-1,\theta_j \in [0,2\pi), r>0\}
\end{eqnarray}
 
Let us assume that there are $\ell$ length constraints where $m<\ell<k$. Recall we have for the constrained size-and-shape coordinates, only three types of spaces, namely 
\begin{enumerate} 
  \item In restricted spherical polar coordinates (for the frame)
  \begin{equation}\label{eq:Restricted}
    \SSS_+^{j-2} \times \RR_+ =\{(x_1,\ldots,x_{j}) \in \RR^{j}:x_1^2 + \ldots +x_{j-1}^2 =1,\; x_{j-1}, x_{j} > 0\}
  \end{equation}
  for $j=1,\ldots,m-1$.
  \item Product of spheres (polysphere) (For fixed length landmarks)
  \begin{equation}\label{eq:polysphere}
    \SSS^j = \{(x_1,\ldots,x_{j+1}) \in \RR^{j+1}: x_1^2 + \ldots +x_{j+1}^2 = 1\},\; j=m,\ldots,\ell,
  \end{equation}
  \item Full Euclidean Coordinates (additional landmarks -- Full Euclidean), essentially $\RR^j$
  \begin{equation}\label{eq:spherical}
    \SSS^{j-1} \times \RR_+ = \{(x_1,\ldots,x_{j},) \in \RR^{j} \}
  \end{equation}
\end{enumerate}

These are the building blocks of the simplex MMP  coordinate system. In restricted spherical polar coordinates given by (\ref{eq:Restricted}), the restriction is $x_{j+1} > 0$ which implies half-circle in the 2D and hemisphere in 3D but on a cylinder as $r>0$ is a variable. There is singularity at $r=0$ which has no real difficulties in constructing plausible pdf on this space. The second set of the coordinates (on sphere) obeys the $\ell$ length constraints where as the left over coordinates are the complete spherical coordinates, In our example of the five landmarks in 3D, $\ell=2$ so there is no the third type. 

We give some particular cases to get insight into constructing plausible distributions for some typical cases and then we can follow the same methods to construct distributions for the general case using \ref{eq:full-sphere-polar}, \ref{eq:half-sphere-polar} in \ref{eq:Restricted}, \ref{eq:polysphere} and  \ref{eq:spherical}.  
 
One of the non-standard case is the distribution on the half-plane mentioned at (\ref{eq:halfplane}) which we rewrite as 
\begin{equation}\label{eq:halfplane1} 
 x=r \cos \theta, y=r\sin \theta,  r>0, 0<\theta<\pi. 
\end{equation}
In this case, we start with bivariate  isotropic normal distribution with random variable $\mb{x}$ (2x1)
$$\mb{x}\sim N(\bm \mu, \Sigma).$$
so $\Sigma=\sigma^2 I$  but  a general $\bm \mu$; the pdf of $(r, \theta)$ is found to be (without  conditioning  the variables)   
\begin{equation}\label{eq:Npdf1} 
\propto r \exp{\{-\kappa_1r^2 +  \kappa_2 r\cos(\theta-\mu)\}}, r>0, 0<\theta<\pi, \kappa_1 >0, -\infty <\kappa_2 <\infty, 0\le \mu \le\pi.  
\end{equation}
We can also write the distribution for the general case with $\bm \mu, \Sigma$ both being general but the form has no  separate term in $r$ in exponent. A better  route  is to start with  a cone in 3D in spherical polar coordinates
\begin{equation}\label{eq:sphCoord}
  r(\sin\theta \cos\phi, \sin\theta \sin\phi, \cos\theta), r>0, 0\le\theta\le \pi, 0 \le\phi\le 2\pi, 
\end{equation}
where for a "cone", $\theta$ is a constant, say c, $0\le c\le \pi$, so the cone is generated by

\begin{equation}\label{eq:coneCorod}
  x= r(\sin c\cos\phi, \sin c \sin\phi, \cos c)^T. 
\end{equation}
Since $r>0$, the cone apex is excluded though for the continuous pdf this has no direct relevance. For a half cone we have $0\le\phi\le \pi.$ 

For the following discussion we replace $ \phi$ with $\theta$. We now start  with trivariate isotropic normal distribution with random variable $\mb{x}$ (3x1),
$$\mb{x}\sim N(\bm \mu, \sigma^2I)$$
with the pdf of $(r, \theta)$ is found on conditioning under the variables \ref{eq:coneCorod} as given in the Theorem below; the distribution has a term in $r$ in the exponent. 
\begin{thm} 
  The pdf of $(r,\theta)$ on conditioning is given by 
  \begin{equation}\label{eq:cone1}
    \{{c(\kappa_1, \kappa_2, \kappa_3)}\}^{-1} r \exp{\{-\kappa_1r^2 + \kappa_2 r + \kappa_3 r\cos(\theta-\mu)\}}, \kappa_1 \in\RR_+, \kappa_2, \kappa_3\in\RR, 0\le \mu \le2\pi, 
  \end{equation}
\end{thm}
where the normalising constant $c(\kappa_1, \kappa_2, \kappa_3)$ is derived below in the Lemma. Note that for $\kappa_2 =0$, we have a bivariate normal in Cartesian coordinates. Further, the conditional distribution of $\theta$ given $r$ is a von Mises distribution. 

Note that the pdf of the cone given by (\ref{eq:cone1}) belongs to the exponential family so we can get the MLE of the following parameters 
$$\kappa_1, \kappa_2, \kappa_3^*= \kappa_3 \cos\mu, \kappa_4= \kappa_3 \sin\mu.$$
However for the moment estimator of the mean is simply $\bar r $ though $E(r)$ has no easy form (except for the concentration case) and a moment estimator of $\mu$ is the sample mean direction. 

\begin{lem}
  The normalizing constant is given by 
  \begin{equation}\label{eq:norm}
    c(\kappa_1, \kappa_2, \kappa_3)= \pi\sum_{m=0}^{\infty} \frac{(\kappa_2)^m}{(\kappa_1)^{\frac{m}{2}+1}} \frac{\Gamma(\frac{m}{2}+1)}{m!}\prescript{}{1}F_1 \left(\frac{m}{2}+1, 1, \frac{\kappa_3^2}{4\kappa_1} \right),
  \end{equation}
  where $\prescript{}{1}F_1(a; b; z)$ is the confluent
  hypergeometric function given by 
  \begin{align}
    \begin{split}
      \prescript{}{1}F_1(a; b; z)&= 1 + \frac{a}{b}\frac{z}{1!} + \frac{a(a+1)}{b(b+1)}\frac{z^2}{2!} + \cdots
      = \sum_{r=0}^{\infty}\frac{(a)_r z^r}{(b)_r r!},
    \end{split}
    \label{eq:1F1sum}
  \end{align}
  where $(a)_r = a(a+1)\cdots(a+r-1)$ and $(a)_0=1$.
\end{lem}

\begin{proof}
  Using the following in the integration of the pdf with respect to $\theta$, we have
  \begin{equation} \label{eq:Bssel}
    \int_0^{2\pi} e^{\kappa \cos(\theta-\mu)} d\theta =2\pi I_0(\kappa)
  \end{equation}
  where $I_p(z)$ is the modified Bessel function of first kind:
  \begin{equation}\label{eq:BesselI}
    I_p(z)= \sum_{m=0}^{\infty} \frac{(z/2)^{2m+p}}{m! \Gamma(p+m+1)}.
  \end{equation}
  we find that 
  \begin{equation} \label{eq:Bssel2}
    c(\kappa_1, \kappa_2, \kappa_3)=
    2\pi \int_0^{2\pi} r \exp{\{-\kappa_1 r^2 + \kappa_2 r\} }I_0(\kappa_3 r) dr.
  \end{equation} 
  Now expanding $\exp{\{ \kappa_2 r\}}$ as a series in (\ref{eq:Bssel2}) and changing the order of integration we have
  \begin{equation} \label{eq:Bssel3} 
    c(\kappa_1, \kappa_2, \kappa_3)=
    2\pi \sum_{r=0}^{\infty} \frac{\kappa_2^m}{ m!}\int_0^\infty r^{m+1} \exp{\{-\kappa_1 r^2 \} }I_0(\kappa_3 r) dr.
  \end{equation}
  We have from \cite{gradshteyn1980}, Section 6.6.31 (Formula 1), 
  $$\int_0^\infty r^{m+1} \exp{\{-\kappa_1 r^2 \} }I_0(\kappa_3 r) dr= \frac{\Gamma(\frac{m}{2}+1)}{2(\kappa_1)^{\frac{m}{2}+1}}\prescript{}{1}F_1 \left(\frac{m}{2}+1, 1, \frac{\kappa_3^2}{4\kappa_1} \right) $$ 
  which when this integral is substituted in (\ref{eq:Bssel3}), we get the normalizing constant given in \eqref{eq:norm}.
\end{proof}

If we restrict $\theta $ to $0<\theta<\pi $, semi-circle then the pdf has the same form except that now $c(\kappa_1, \kappa_2, \kappa_3)$ to be replaced by $c(\kappa_1, \kappa_2, \kappa_3)/2$ since at the step (\ref{eq:Bssel}) use now 
\begin{equation} \label{eq:Bssel4}
    \int_0^{\pi} e^{\kappa \cos(\theta-\mu)} d\theta =\pi I_0(\kappa)
\end{equation}
which is half of the integral at (\ref{eq:Bssel}). In fact, this is the case with

\begin{equation}\label{eq: Resrticted 2D }
\RR_+^{2} = \{(x_1,x_2) \in \RR^{2}:x_1^2 +x_2^2 =r_2^2,\; r_2, x_{2} > 0\}.
\end{equation}

We can carry out inference in the standard way for a member of the exponential family. However, as the normalising constant is complicated, one can use the score matching estimation (see, for example,).

The other type of random variables in the coordinates are on the full sphere so we can use the Fisher distribution. Let $\ell$ be the spherical coordinates then we can use
\begin{equation}\label{eq:Fisher1} 
    f(\ell, \mu)\propto e^{\kappa \ell^T\mu}\sin\theta
\end{equation}
where we have $\ell,\mu$ are on the sphere and $\kappa>0$.

We can construct the distributions as joint distribution of the coordinates in various directions. The simplest case is the distribution of $r,\theta \ell$ where we restrict $\theta $ to $ 0<\theta<\pi $ in (\ref {eq:coneCorod}) then from (\ref{eq:cone1} ) and (\ref{eq:Fisher1}) is given by 
\begin{equation}\label{eq:cone1Sphere}
 \propto r e^{\{ \kappa \ell^T\mu-\kappa_1r^2 + \kappa_2 r + \kappa_3 r\cos(\theta-\mu) \}}, 
\end{equation}
where the ranges of the parameters are given above at (\ref{eq:cone1} ) and (\ref{eq:Fisher1}).

For the cone type coordinates, we can extend the work to any dimension $p$ by first writing the $p$ spherical coordinates (\ref{eq:coneCorod}) and then follow the same conditioning as above, and for the poly-sphere we can use the extension of the Jupp and Mardia distribution (\cite{Jupp-Mardia80}) given in \cite{mardiaetal2016ar} for two-spheres as the joint Fisher distribution and then to poly-sphere. Thus we have a rich class of new distribution by using their joint distribution. Note that Appendix 1 shows the distributions becomes  intricate when we use the standard coordinates system  for the length constrained size-and-shape analysis.

\section{Outlook: Statistical Clustering of Constrained Shape}\label{scn:clusteringOfShape}

One application of constrained shape analysis lies in the analysis of geometrical structure of biomolecules, in particular, in classifying possible shapes. To that end, in a previous publication \citep{mardia2021principal} we introduced the MINT-AGE algorithm that was used, among others for reconstruction of suites of SARS-CoV-2 in \citet{wiechers2023learningJRSSC}. One ingredient in MINT-AGE is nonparametric mode hunting on a dominating first principal circle. In an ongoing collaboration with the Richardson Lab at Duke University application, some highly relevant conformation classes studied are too rare, however, to be found with this nonparametric method. Here we illustrate a more powerful parametric extension, designed to find very small clusters with up to 3 elements only as well. To this end, first, we briefly recall MINT-AGE in the next section.

While this new method was already successfully employed for constrained shape modelled by dihedral angles \citep{wiechers2025RNAPrecis}, in future work we plan to employ MINT-AGE, including parametric mode hunting, to constrained shape modelled by multicentred MPP coordinates, involving a polysphere part.

\subsection{MINT-AGE Overview}

In its original version \citep{mardia2021principal}, MINT-AGE was applied to constrained shape of biomolecules represented by dihedral angles, i.e. torus-valued data, building on \emph{torus-PCA} from \cite{EltznerHuckemannMardia2018}. It is a multistep clustering method and in its first step it iteratively adapts to clusters from low density to high density. In the second step it maps cluster obtained from the first step -- as they lie on a torus, not allowing for straightforward further dimension reduction -- to a suitably stratified sphere. On that stratified sphere, with a variant of \emph{principal nested spheres} (PNS) from \cite{jung2010analysis}, the dominating first principle circle is determined, onto which the cluster is projected. This cluster is then subjected to nonparametric mode hunting extending \cite{duembgen2008} to the circle and using the WiZer from \cite{WiZer2016} in order to detect subclusters.

Notably, torus-PCA extends naturally to polysphere-PCA \citep{HE_LASR15} and thus, the dimension reduction and mode hunting steps of MINT-AGE naturally apply to constrained shape involving polysphere coordinates.

\subsection{Parametric Mode Hunting on the Circle for Very Small Cluster Sizes}

The circular mode hunting as described in \citet{wiechers2023learningJRSSC} relies on a multiscale test which does not posit a specific model for the shape of modes. The price for the model agnosticism is a low power of the methods so that clusters which comprise less than around $30$ data points are not distinguished.

In order to achieve higher power, we consider a simplified model based approach. We fit two models to the one dimensional projected data on the circle:
\begin{description}
  \item[Model I:] One single wrapped normal distribution with variable $\mu$ and $\sigma$.
  \item[Model II:] Two wrapped normal distributions with variable $\mu_1$, $\sigma_1$, $\mu_2$ and $\sigma_2$.
\end{description}
We use an EM algorithm (see, for example, \cite[384]{MardiaKentTaylor2024}) for the second model. Then, we apply a likelihood ratio test with confidence level $\alpha \in (0,1])$ for the estimated models as follows. Let $\ell_{I}$ be the log-likelihood according to the fit of Model I and $\ell_{II}$ be the log-likelihood according to the fit of Model II. Then obviously $\ell_{II} \ge \ell_{I}$. According to Wilks' Theorem \citep{Wilks1938},
\begin{align*}
    2 (\ell_{II} - \ell_{I}) \inD \chi^2_2,
\end{align*}
for sample sizes $\to \infty$, where the of degrees $2$ of freedom of the $\chi^2$ distribution stem from the fact that Model II has two more degrees of freedom than Model I. Thus, the likelihood ratio test is a $\chi^2$ test.

If the test does not reject, the data set is added to the final clusters list. If the test rejects, we assign to each point the label of the mixture component for which that point has a higher likelihood value and partition the data set according to these labels. The procedure is then repeated on the two resulting partitions, until none of the likelihood ratio tests rejects.

Here is a pseudo code for the new parametric mode hunting. 
\begin{enumerate}
  \item INPUT: projections of a data cluster onto a one-dimensional sphere with confidence level $ \alpha $ (typically $0.05$) and minimal cluster length $\kappa$ (typically $3$).
  \item Fit the two Gaussian mixture Models I and II above with 1 and 2 components, respectively, to the one dimensional projections and calculate the log-likelihood for both fits.
  \item Assign the subclusters for both fits and check if they have at least $ \kappa $ points; else don't split cluster and RETURN cluster.
  \item Perform the above likelihood ratio test to the level $\alpha$. 
  \item If the test rejects, split the cluster into the two subclusters and GOTO Step 2 for each subcluster separately; else don't split cluster and RETURN cluster.
\end{enumerate}

\section{Discussion}

We can extend these size-and-shape coordinates in various directions, For example, we can consider the full group of Euclidean transformations, namely also including reflection, leads, allowing $R \in O(m) = \{R \in \RR^{m\times m}: R^TR = I_m\}$, to \emph{size-and-reflection-shape coordinates} also not pursued in this work (see \citet{Dryden2016}).
Extension to similarity shape, affine shape etc also can be done.
These coordinates are already developed with QR decomposition in \citep{goodallkvm1991,goodallkvm1992,goodall1993multivariate}.

Motivated by our applications, we have given data analysis only for concentrated coordinates. Our framework, however, allows to construct distributions and carry out inference even when the distribution of the coordinates is not concentrated. These model constructions allow us to have unsupervised learning for clustering. We aim to use these models on combining Cluster 1 and Cluster 2 (P33, see Section \ref{Summary}) for testing our future unsupervised learning.

\section*{Acknowledgments}

The authors would like to thank Franziska Hoppe and Henrik Wiechers for helpful discussions. BE and SH acknowledge support from DFG SFB 1456; SH acknowledges support from DFG HU-1575-7, DFG GK 2088 and the Niedersachsen Vorab of the Volkswagen Foundation. KVM would like to thank Ian Dryden and Colin Goodall for the helpful comments.

\appendix

\section{Appendix 1: Mardia Size and Shape Coordinates in 3D}

We start with the Bookstein coordinates in 2-D. Let $x_i, i=1,\ldots,k$ be the original landmarks. For 2-D, the first two new coordinates are $(0, 0)$ and $(d, 0)$ using $x_1$ and $x_2$. That is
\[
u_i = S(\phi ) y_i, \;\;\; y_i = x_i - x_1, \;\;\; i = 1, \ldots, k,
\]
where now $S(\phi)$ is a $2 \times 2$ rotation matrix 
\[
S(\phi) = \left[ 
\begin{array}{cc}
\cos\phi & \sin\phi \\
-\sin\phi & \cos\phi 
\end{array}
\right]
\]
such that 
\[
u^T_2 = (|x_2 - x_1|, 0) = (|y_2|, 0).
\]
Following \citet{Mardia09} and \citet{Dryden2016} we indicate how to construct Bookstein-type coordinates for size-and-shape (form) analysis. Let $X(k \times 3)$ be the configuration matrix with rows $x_i$. First note that there are six unknowns in the rotation matrix $A$ (3 Eulerian angles) and a translation vector, but $x_1$ and $x_2$ are not enough to determine the form coordinates, since the distance between $x_1$ and $x_2$ is fixed,  so we have to use $x_3$ for the one remaining constraint. We first use the point $x_1$ as the origin so that
\begin{equation} \label{yCoord}
  y_i = x_i - x_1, \;\;\; i = 1, \ldots, k.
\end{equation}

Then use $y_2$ (i.e. $x_2$) to fix the co-latitude and longitude of the points, i.e. let $(\theta_2, \phi_2)$ be the polar coordinates of $y_2$ ($z$-axis is the north pole),
\begin{equation} \label{y2Coord}
  y_2 = |y_2|\left[
  \begin{array}{c}
    \sin\theta_2 \cos\phi_2 \\
    \sin\theta_2 \sin\phi_2 \\
    \cos\theta_2
  \end{array}
  \right]
\end{equation}

So that we have the new coordinate system (Euler angles) 
\begin{equation} \label{uCoord}
  u_i = R(\theta_2, \phi_2) y_i, \;\;\; i=1, \ldots, k
\end{equation}
where
\begin{equation} \label{3Drot}
  R(\theta, \phi) = \left[
  \begin{array}{ccc}
    \cos \theta \cos\phi & \cos\theta \sin\phi & -\sin\theta \\
    -\sin\phi            & \cos \phi & 0 \\
    \sin\theta \cos\phi  & \sin\theta \sin\phi & \cos\theta
  \end{array}
  \right].
\end{equation}
That is 
\begin{equation} \label{u1u2}
  u_1=(0,0,0)^T, u_2= |y_2|(0,0,1)^T.
\end{equation}
Now rotate the new coordinates around the coordinates $u_3$ (or $x_3$), i.e. if
\begin{equation} \label{u3Coord}
  u_3 = \left[
  \begin{array}{c}
    a \\ b \\ c
  \end{array} \right] =|y_3| \left[
  \begin{array}{c}
    \sin\theta_3 \cos\phi_3\\
    \sin\theta_3 \sin\phi_3 \\
    \cos\theta_3
  \end{array}
  \right],
\end{equation}
then the Bookstein type coordinates for the form in 3-D (angle $\phi_3$ with new $x$-axis) are
\begin{equation}
  v_i = S(\phi_3 ) u_i, \;\;\; i = 1, \ldots, k,
\end{equation}
where
\begin{equation} \label{2Drot}
  S(\phi ) = \left[
  \begin{array}{ccc}
    \cos\phi & \sin\phi &0\\
    -\sin\phi & \cos\phi &0\\
    0 & 0 & 1 
  \end{array}
  \right].
\end{equation}
Thus we have
\begin{equation} \label{vAngle3}
  v_1=(0,0,0)^T, v_2=|y_2|(0,0,1)^T, v_3=|y_3|(\sin \theta _3, 0, \cos \theta _3)^T.
\end{equation}
and 
\begin{equation} \label{vAngles}
  v_i = S(\phi_3 ) R(\theta_2, \phi_2) y_i, \;\;\; i = 1, \ldots, k, 
\end{equation}
where $y_i = x_i - x_1, \;\;\; i = 1, \ldots, k$. Further, we have 
\begin{equation} \label{vCoord}
  v_i= |y_i|S(\phi_3 ) R(\theta_2, \phi_2) \left[
  \begin{array}{c}
    \sin\theta_i \cos\phi_i \\
    \sin\theta_i \sin\phi_i \\
    \cos\theta_i
  \end{array}
  \right],
  \;\;\; i = 1, \ldots, k,
\end{equation}
where $ 0\leq \theta_i\leq \pi, 0\leq \phi_i\leq 2\pi,\; i = 3, \ldots, k.$

Thus we have a representation of Bookstein Coordinates (\ref{vCoord}) all in terms of distances and angles, and can call it ``polar representation'' versus ``Euclidean representation''. This is more useful in Bioinformatics where $\theta$ is the bond angle and $\phi$ is the dihedral angle. We can extend these to any dimension.

In Bioinformatics, these coordinates are termed bond-angle-torsion (BAT) co-ordinates, see for example, \citet{Killianetal07} and can be given a tree structure.

\subsection{Distributions and Mardia's Constrained Size-and-Shape Coordinates}\label{DistrB}

In 2D, a full distribution has been given in \cite{MR1158517} and for concentrated data, angular representation can be worked out. The approach adopted is to integrate out ``nuisance'' variables under normality with full mean vector and covariance matrix for the whole configuration of the original landmarks $x's$. Better approach is perhaps to use conditional distributions. How to use the representation given by \ref{vCoord} to find a plausible distribution of 
$$v_i = S(\phi_3 ) R(\theta_2, \phi_2) y_i, \;\;\; i = 2, \ldots, k,$$ 
assuming that the first three landmarks $x_i, i=1,2,3$ selected are concentrated and are well separated. 
 
Consider k=5 with two lengths fixed. Consider the size-and-shape coordinates given by (\ref{vAngle3}). Suppose the first length fixed is between $x_2$ and $x_3$ then we can ignore $v_2$ as $|y_2|= |x_2-x_3|=1, \text{say}$ is fixed so the coordinates are now simply 
\begin{equation} \label{vAngleR}
  v_3=(\sin \theta _3, 0, \cos \theta _3)^T,v_4 = S(\phi_3 ) R(\theta_2, \phi_2) y_4,\; v_5 = S(\phi_3 ) R(\theta_2, \phi_2) y_5
\end{equation}
Now suppose another length between $x_4$ and $x_5$ is fixed then with $y_4= x_4-x_3, y_5= x_5-x_3$ so we have a constrain
$$|y_4-y_5|= |v_4-v_5|$$
so we have the system
\begin{equation} \label{vAngle5}
  v_3=|y_3|(\sin \theta _3, 0, \cos \theta _3)^T, v_i = S(\phi_3 ) R(\theta_2, \phi_2) y_i, i=4,5
\end{equation}
under constraint
$$|y_4-y_5|= |v_4-v_5|=|x_4-x_5|. $$ 
Now for example consider, the joint distribution of $v_4,v_5$ which depend on $x_1, x_4,x_5$ where very intricate with heavy spurious dependence in contrast to the Simplex coordinates where the distribution of the coordinates depends only on the fixed landmarks only so under the isotropic normal model, the simplex system will be more ``stable''.

\bibliographystyle{rss}
\bibliography{ref_PDWC,shape,circular,stats,proteins,richardsonrefs}

\end{document}